\documentclass[conference]{IEEEtran}

\IEEEoverridecommandlockouts
\IEEEpubid{\makebox[\columnwidth]{
	979-8-3195-3997-7/26/\$31.00~\copyright2026
	IEEE \hfill} \hspace{\columnsep}\makebox[\columnwidth]{ }
}

\usepackage{amssymb}
\usepackage{amsthm}
\usepackage{mathpartir} 
\usepackage[dvipsnames]{xcolor} 
\usepackage{amsmath}
\usepackage{multicol}
\usepackage{fancyvrb}
\usepackage{hyperref}
\usepackage{cleveref}
\usepackage{xspace}

\newtheorem{theorem}{Theorem}
\newtheorem{prop}{Proposition}

\newtheorem{coro}{Corollary}
\newtheorem{definition}{Definition}

\newcommand{\NN}{\mathbb{N}}
\newcommand{\id}{\mathrm{id}}

\newcommand\gramdef{\ ::= \ } 
\newcommand{\nt}[1]{\textit{#1}} 
\renewcommand{\mid}{\ |\ }
\newcommand{\gramspace}{-0.25em} 
\newcommand{\code}[1]{\textbf{\texttt{#1}}} 
\newcommand{\cspace}{\;} 
\newcommand{\rank}[1]{rank(#1)}

\newcommand{\node}[3]{\code{node}\cspace #1\code{(}#2\code{)=}\code{(}#3\code{)}} 

\newcommand{\when}[2]{#1\cspace\code{when}\cspace#2}

\newcommand{\fby}[2]{#1\cspace\code{fby}\cspace#2}
\newcommand{\merge}[3]{\code{merge}\cspace#1\cspace#2\cspace#3}
\newcommand{\annot}{\code{:\!:}}
\newcommand\inst[1]{\code{inst} {#1}}

\newcommand{\base}[1]{base(#1)}
\newcommand{\ebase}[1]{ebase(#1)}
\newcommand{\subs}[1]{subs(#1)}

\newcommand{\cnot}[1]{\code{not}\cspace #1} 
\newcommand{\ceq}[2]{#1\cspace\code{=}\cspace #2;} 

\newcommand{\Tick}{\code{.}}

\newcommand{\on}[2]{#1\cspace\code{on}\cspace#2}

\newcommand{\var}[1]{\!\code{`}#1} 

\newcommand\cclus{{\em CC-lus}}
\newcommand\ccnew{{\em CC-new}}

\newcommand{\figref}[1]{Fig.~\ref{fig:#1}}

\newcommand{\secref}[1]{Section~\ref{sec:#1}}

\newcommand{\thmref}[1]{Theorem~\ref{thm:#1}}
\newcommand{\lemmaref}[1]{Lemma~\ref{lemma:#1}}
\newcommand{\propref}[1]{Proposition~\ref{prop:#1}}

\newif\ifextended 

\extendedtrue 

\newtheorem{ext-theorem}{Theorem}[theorem]
\newtheorem{ext-prop}{Proposition}[prop]
\newtheorem{ext-lemma}{Lemma}[lemma]
\stepcounter{lemma} 
\newtheorem{ext-coro}{Corollary}[coro]
\newtheorem{ext-definition}{Definition}[definition]

\title{
  Relaxed activation analysis of dataflow networks\\
  \Large A clock calculus for machine learning and real-time scheduling
}
\author{
    \IEEEauthorblockN{William Gaudelier}
    \IEEEauthorblockA{Inria, Paris, France
    \href{mailto:william.gaudelier@inria.fr}{william.gaudelier@inria.fr}\\}
    \and
    \IEEEauthorblockN{Albert Cohen}
    \IEEEauthorblockA{Google DeepMind, Paris, France}
    \and
    \IEEEauthorblockN{Dumitru Potop Butucaru}
    \IEEEauthorblockA{Inria, Paris, France}
    \thanks{Partially funded by the CHIPS JU Shift2SDV project. The authors would also like to thank the anonymous reviewers for their valuable comments.}
}

\IEEEspecialpapernotice{\vspace*{-7mm}}
\begin{document}

\maketitle

\begin{abstract}
  Previous work has shown that
  the simple dataflow primitives of the Lustre language allow the
  natural, semantically unambiguous, and compact representation of
  machine learning (ML) applications, including models featuring
  complex conditional execution and recurrent state.
  The Lustre \emph{clock calculus} is responsible for the static determination
  of important properties such as liveness (absence of deadlocks)
  and static memory bounds.
  Yet existing clock calculi are tailored for embedded control applications.
  We show they do not cater for the representation of
  control patterns commonly found in training algorithms, resulting in
  cumbersome expressions and inefficient compilation.
  We propose a conservative extension of Lustre's clock calculus
  addressing this limitation, thereby facilitating the embedding of
  ML models in reactive applications.
\end{abstract}

\section{Introduction\label{sec:intro}}

From embedded control theory and signal processing to data science and
machine learning (ML), dataflow formalisms are used for both high-level specification and model exchange, in formalisms such as Simulink
\cite{simulink}, Lustre \cite{lustreRTSS,Lustre,Bal08}, or \href{https://onnx.ai}{ONNX} \cite{onnx}.
They are also used as support of lower-level optimization and resource
allocation algorithms, like in Apache Flink \cite{flink},
TensorFlow graphs \cite{tfgraph}, or JAX and XLA's
\href{https://openxla.org/xla/operation_semantics}{HLO} \cite{xla}.

The dataflow paradigm is intrinsically concurrent. Within it, large
subclasses of \emph{reactive} and \emph{deterministic} dataflow
applications have been identified \cite{kahn1974semantics}
that are amenable to efficient implementation in multiple domains. However,
while determinism is desirable, it is not sufficient to ensure safe
implementation in an embedded context. In particular, it is difficult
to ensure the absence of either {\em deadlocks} or {\em unbounded
  accumulation} of values (memory leaks).

In TensorFlow graphs such activation errors are observable as {\em
  unspecified behavior} \cite{tfgraph} or {\em silent errors}
\cite{tambon2021silent}. Various {\em well-formed graph patterns} have
been proposed \cite{yu2018dynamic}, but no general analysis exists,
making the design and debugging of large graphs difficult
\cite{cai2016tensorflow} and limiting adoption. Other studies have
identified classes of dataflow applications which support scalable
deadlock/leak analysis. In Synchronous Data Flow (SDF)
\cite{GeilenBasten2010KPN}, the analysis is based on conservation laws
ensuring that the production and consumption of dataflow messages balance
each other during execution, but expressiveness is limited.
In this paper we consider the broader class of \emph{reactive, dataflow synchronous} systems \cite{CaP96}.
These found extensive use in control automation and machine learning, through
formalisms such as Simulink, Lustre, or ONNX.

\subsection{Cycles, activation, and clock calculus\label{sec:activation}}

In dataflow synchronous systems, execution is organized as a {\em
  sequence} of {\em execution cycles}. Within each execution cycle,
each dataflow node can be executed at most once, and each variable can
be assigned at most one value. Thus, the {\em activation} of every
computation and the {\em presence} of every variable in each cycle can
be represented with a stream (sequence) of Booleans indexed by the
indices of the cycle.

Symbolic representations of these streams of Booleans, known as {\em
  logical clocks}, allow determining that the dataflow model cannot
deadlock or accumulate data {\em within any execution cycle}. To provide
an intuition, consider the following Lustre {\em equation}
(statement), which is syntactically correct:

{\scriptsize\begin{verbatim}
     y = x + (x when c);
\end{verbatim}}

\noindent
It specifies that \texttt{y} is computed as the sum of \texttt{x}
and the {\em subsampling} of \texttt{x} on the Boolean condition
\texttt{c}.  At cycles where \texttt{x} has a value but \texttt{c} is
false, the second term of the addition is {\em absent} and the
computation ``diverges'': under dataflow semantics, the value of the
left-hand-side operand \texttt{x} of \texttt{+} must be stored into a
queue, possibly leading to an infinite accumulation over execution
cycles. Dataflow synchronous semantics outlaw this behavior by
requiring both inputs of \texttt{+} (or any function) to be
present (resp.\ absent) in a given cycle.
Detecting such errors can be done at
runtime, as in TensorFlow dataflow graphs, but languages such as
Lustre require, for safety and efficiency purposes, that the compiler
rejects such programs by means of low-complexity static analysis.

To this end, the programmer or compiler attach activation information to every dataflow variable in the form of a type; this ``activation'' type is known as a \emph{clock}. 
In our example, assuming that \texttt{x} has a value at every cycle, it
has clock ``\texttt{.}'', also called the {\em base clock}.
Then, the expression ``\texttt{x when c}'' will have clock ``\texttt{. on c}'', meaning that it has a value on all cycles where \texttt{c} is true.
Finally, the \texttt{+} operator requires all its inputs to have the same clock, which is not the case here, resulting in a type mismatch.
A type system providing such an activation analysis is
called a {\em clock calculus} \cite{Bal08}.



Lustre clocks have the form
``{$\on{\on{\on{\on{\Tick}{c_1}}{c_2}}{\ldots}}{c_n}$}'' where
$c_i$ is either \texttt{x}$_i$ or ``\texttt{not x}$_i$'' for some
Boolean variable \texttt{x}$_i$.
%
A computation $f$ having clock ``$\on{\on{\Tick}{a}}{b}$'' means that,
to activate $f$ at a given cycle, control first tests the value of
$a$. If the test succeeds variable $b$ may then be tested; and when $b$ is
also true $f$ is activated.
To allow this, $a$ must be present at all cycles (its clock must be
``$\Tick$''), and $b$ must be present at cycles where $a$ is true.
To ensure that programs are interpretable as Kahn process networks \cite{CaP96}, Lustre also requires that $b$ is {\em only} present at cycles where $a$ is
true---its clock must be ``$\on{\Tick}{a}$''.
Now consider a variable $c$ present at every cycle.
Expression ``$\on{\on{\Tick}{c}}{a}$'' is invalid since $a$ is not present exactly when $c$ is true.
To subsample on $a$ after subsampling on $c$, one
must introduce a new variable, e.g.\ $e$, defined as
``$\when{a}{c}$'', which is equal to $a$ when $c$ is true and is
absent otherwise. Then, the clock $ck_1=$~``$\on{\on{\Tick}{c}}{e}$''
identifies cycles where $c$ and $a$ are both tested, in
order. Symmetrically, by setting ``$f=\when{c}{a}$'' one may define the
clock $ck_2=$~``$\on{\on{\Tick}{a}}{f}$'' which first tests
$a$. Obviously, $ck_1$ and $ck_2$ are true in the same execution
cycles. However, statically determining their equality
requires interpreting the subsampling relations, which the clock calculus
of Lustre does not allow.

\begin{figure}[t]
  \vspace*{-4mm}
  {\small\begin{align*}
    \nt{prog}       &\gramdef\ \nt{node} \mid \nt{node}~\nt{prog} \\[\gramspace]
    \nt{node}       &\gramdef\ \node{\nt{id}}{\nt{id-ck-list}}{\nt{id-ck-list}} \nt{eqs} \\[\gramspace]
    \nt{id-ck-list} &\gramdef\ \nt{id-ck}\code{,}\ldots\code{,}\nt{id-ck}\\[\gramspace]
    \nt{id-ck}      &\gramdef\  \nt{id} \mid \nt{id} \color{red}~\annot~\nt{ck}\\[\gramspace]
    \nt{eqs}        &\gramdef\ \nt{eq} \mid \nt{eq}~\nt{eqs} \\[\gramspace]
    \nt{eq}         &\gramdef \ceq{\nt{id-ck}}{\nt{e}} \mid  \ceq{\nt{id-ck-list}}{\inst~{\nt{id}}\code{(}\nt{e}\code{,}\ldots\code{,}\nt{e}\code{)}}  \\[\gramspace]
    \nt{e}          &\gramdef \nt{id} \mid \nt{k} \mid \nt{id}\code{(}\nt{e}\code{,}\ldots\code{,}\nt{e}\code{)} \mid \nt{e}~\nt{op}~\nt{e} \\[\gramspace]
                    &\phantom{\gramdef\ } \mid \fby{\nt{k}}{\nt{e}} \mid \when{\nt{e}}{\color{red}\nt{sse}} \color{black} \mid 
    \merge{\color{red}\nt{sse}\color{black}}{\nt{e}}{\nt{e}} \\[\gramspace]
    \color{red}\nt{sse} &\color{red}\gramdef \nt{id} \mid \cnot{\nt{id}} \\[\gramspace]
    \color{red}\nt{ck}  &\color{red}\gramdef \color{red}\Tick  \mid \on{\nt{ck}}{\nt{sse}} \color{blue}\mid \var{\nt{id}}
  \end{align*}}
\vspace*{-7mm}
\caption{Syntax of Lustre.
  \secref{newclocks1} presents the extended clock syntax (in red).
  Clock variables (quoted, in blue) only occur during clock inference.
  Operators $\nt{op}$ ($\mathtt{+}$) and constants $\nt{k}$
  ($\mathtt{true}$, $\mathtt{42}$, constant tensors) are not detailed.
  \label{fig:lustre}}
\vspace*{-4mm}
\end{figure}

\subsection{Contribution\label{sec:contribution}}
The traditional work around is to adapt the specification itself:
when a variable with clock
``$\on{\on{\Tick}{a}}{f}$'' must be used on clock
``$\on{\on{\Tick}{c}}{e}$'', it is first converted to a variable on
clock ``$\Tick$'' by padding absent values with some default value.
This reduces the efficiency of real-time scheduling \cite{emsoft09}
and breaks modularity.
Furthermore, as we shall
see in \secref{motivation}, the representation of backpropagation
training in ML systematically raises this very problem.

We thus extend the Lustre clock calculus, henceforth named \cclus, to provide additional ordering freedom when subsampling Boolean variables, when presence permits.
The resulting clock calculus \ccnew\ is a
conservative extension of \cclus: existing clock types
are preserved with their semantics unchanged. New clocks incorporate
the aliasing and subsampling information to allow the desired
specification style and low-complexity type inference. The clocks
defined above are respectively represented as
$ck_1=$``$\on{\on{\Tick}{c}}{(\when{a}{c})}$'' and
$ck_2=$``$\on{\on{\Tick}{a}}{(\when{c}{a})}$''.
We proved the correctness of \ccnew and we implemented it in a prototype Lustre compiler.

\subsection{Outline}
The rest of the paper is structured as follows: \secref{lustre}
introduces the Lustre language and (in \secref{motivation}) a
motivating ML example. \secref{related} discusses related
work. \secref{newclocks1} defines the extended clock calculus
and states fundamental correctness results. \secref{conclusion}
concludes.

\section{Dataflow programming in Lustre\label{sec:lustre}}

\subsection{Overall program structure}
Lustre is the prototypical dataflow synchronous language.
It has been the subject of extensive formal
analysis \cite{lushist}, with only 4 {\em
  primitive statements} \cite{Bal08} from which all others can be
derived. Our analysis naturally covers only this core language.
The core syntax is provided in \figref{lustre}.
A Lustre program consists of a sequence of {\em nodes}, each defining a {\em dataflow graph}. The program in \figref{conv} features 2 nodes.
The first implements a convolutional layer;
the second instantiates it twice to form a simple ML model.
Each node has an {\em interface} defining its name, inputs and
outputs, followed by one or more semicolon-terminated {\em equations}
(non-terminal $\nt{eq}$) that define the dataflow graph.

\begin{figure}[t]
{\scriptsize\begin{verbatim}
node conv_layer(x)=(y)
  kernel = [[1.,2.],[3.,4.]]; bias = 0.;
  y = relu(conv2d(x,kernel) + bias);
\end{verbatim}}
{\scriptsize\begin{verbatim}
node model(x)=(y)
  z = inst conv_layer(x);
  y = inst conv_layer(z);
\end{verbatim}}
\vspace*{-3mm}
  \caption{Simple convolutional layer and model in Lustre\label{fig:conv}}
\vspace*{-3mm}
\end{figure}

Library functions such as \verb|relu| or \verb|conv2d| are directly
accessible as long as they are {\em pure} and that they always terminate.
To be activated in a
cycle, a function call must receive a value on every input. When this
happens, it performs its computation and produces a value on its
output.

The definition of a node can use other Lustre nodes through a process
called {\em instantiation}.
In \figref{conv}, the \verb|model| node instantiates node
\verb|conv_layer| twice. The {\em semantics} of instantiation is
faithfully defined by inlining.
A node cannot directly or indirectly instantiate itself.

\subsection{Stateful programs \label{sec:stateful}}
\figref{conv} showcased the modeling of feed-forward networks
\cite{goodfellow2016deep} (e.g.\ convolutional, residual) where
the computations of one execution cycle are independent of those of
other cycles. 
While such networks are preferred for low-level
perceptual tasks, many higher-level activities require a temporal
context \cite{goodfellow2016deep}, e.g.\ to determine the speed or
intent of a pedestrian or car, or to remember a speed limit.
Statefulness is also required, even
in feed-forward networks, when the model incorporates scheduling
(e.g.\ pipelining) \cite{rybakov2020streaming}.
Stateful behavior in deep learning typically involves a
{\em recurrent neural network (RNN)}
\cite{goodfellow2016deep}. Recurrent layers range from the simple LSTM (Long Short-Term Memory) and GRU (Gated Recurrent Unit), to SSMs (State Space Models) such as Mamba \cite{gu2024mamba},
a low-complexity alternative for transformers.

To represent recurrences/state Lustre uses the $\code{fby}$ dataflow
primitive. In the first cycle where equation
``$\mathtt{y}\;=\;k\;\code{fby}\;\mathtt{x}$'' is active, the output
\verb|y| is assigned the constant $k$. In every subsequent cycle where
the \code{fby} is active, \verb|y| is assigned the value of \verb|x|
of the {\em previous} activation cycle. In other terms, $\code{fby}$
defines a first-order recurrence relation, just like a {\em unit
  delay} \cite{astrom1997computer} does in control theory. Using
\code{fby}, the Lustre representation of generic RNN layer
is:

{\scriptsize\begin{verbatim}
  node rnn_layer(x)=(y)
    y,next_state = inst cell_layer(x,state);
    state = init_state fby next_state;
\end{verbatim}}

\noindent
At each cycle, the RNN layer uses another node (traditionally called
the {\em cell} layer) to compute the current output and next state
from the current input and current state. The next state is propagated
to the next cycle using the $\code{fby}$ primitive.

A node can contain, directly or as part of instantiated nodes,
multiple $\code{fby}$ primitives. Each one becomes a separate state
element. Note that the state is internalized. Exposing it as a tuple
or some other data structure is not necessary, unlike functional ML
languages such as \href{https://github.com/jax-ml/jax}{JAX}. This
supports a more natural and modular specification style.

\subsection{Conditional control}
In the dataflow paradigm, computation takes place whenever inputs are ready.
Thus, for a function $f$ to execute only at cycles where a condition $c$
is true, one must make sure that $f$ only
receives its inputs at cycles where $c$ is true. This is the purpose of the
{\em subsampling} dataflow primitive $\code{when}$. The
equation ``$\mathtt{x}\;=\;\mathtt{y}\;\code{when}\;\mathtt{c}$''
defines \verb|x| to the value of \verb|y| at cycles where \verb|c| is
true, and leave it {\em absent} in all other cycles.

The $\code{merge}$ primitive does the opposite operation. The equation
``$\mathtt{x}\;=\;\code{merge}\;\mathtt{c}\;\mathtt{y}\;\mathtt{z}$''
takes as input a Boolean \verb|c| and two data inputs, the first
present when \verb|c| is true, and the second when \verb|c| is
false. At cycles where \verb|c| is true, the output \verb|x| takes its
value from \verb|y|. At cycles where \verb|c| is false, \verb|x| takes
its value from \verb|z|.
\begin{figure}[t]
{\scriptsize\begin{verbatim}
  node gmoe(x)=(y)
    w1,w2 = inst gating(x);
    c1 = (w1 = 0); c2 = (w2 = 0);
    y1 = inst expert1(x when c1);
    y2 = inst expert2(x when c2);
    y = (merge c1 y1 0)*w1 + (merge c2 y2 0)*w2;
\end{verbatim}}
  \vspace*{-3mm}
  \caption{Minimal gated mixture of experts layer in Lustre\label{fig:gmoe}}
  \vspace*{-3mm}
\end{figure}

With these two primitives, we can encode the Gated Mixture of Experts
(GMoE) layer \cite{shazeer2017outrageously} which enables the
efficient execution of LLMs by activating expert networks
on demand, as determined by a {\em gating} network. In \figref{gmoe},
we provide a minimalist version of GMoE with two experts.  Notice how
\verb|x| is subsampled on conditions \verb|c1| and \verb|c2| to ensure
that the two experts are conditionally executed.


\subsection{The Lustre clock calculus \cclus}\label{sec:cclus}


The intuition for Lustre's clock calculus has been introduced
in \secref{activation}.
Let us formalize it as a dependent type system of logical clocks.
Different (compatible) formalisms exist
\cite{Bal08,CaP96,CoP03}; here we follow \cite{Bal08}.

\subsubsection{Syntactic aspects}
The only Lustre statement that can generate a new activation condition/type is
$\code{when}$. It takes as input an expression to be subsampled and a
predicate, also called subsampling expression (non-terminal $\nt{sse}$
in \figref{lustre}). To avoid the complexity trap of general 
satisfiability, Lustre drastically limits the expressiveness of
predicates, which can only use Boolean variables and the $\code{not}$
operator.

The syntax of clocks (non-terminal $\nt{ck}$) closely follows that of
subsampling. Starting from the base clock ``$\Tick$'' defining
activation at every cycle, it allows progressively subsampling on
multiple predicates.
To support Hindley-Milner type inference, this basic syntax is
extended with {\em clock variables} of the form $\var{\mathit{id}}$
(in blue in \figref{lustre}). This allows, during inference, the
manipulation of clocks of the form
``$\on{\on{\on{\var{b}}{c_1}}{\ldots}}{c_n}$''
where $\var{b}$ is a clock variable.

\subsubsection{Clock constraints\label{sec:clkconstraints}}
Clock inference fundamentally amounts to constraint solving:
assign a clock to every variable and expression
while satisfying {\em clock constraints} coming
from 2 sources: {\em explicit clock annotations} attached by the
programmer (non-terminal $\nt{id-ck}$) and {\em
  implicit clock constraints} associated to the various dataflow
primitives and to the clocks themselves. 
Recall from \secref{intro} the implicit clock
consistency constraint: if $ck$ is a clock and $c$ a Boolean variable,
then for ``$\on{ck}{c}$'' to be well-defined we require that the clock
of $c$ is $ck$. Similarly, if $e$ is ``$\when{e^\prime}{c}$'' and the
clock of $c$ is $ck$, then the clock of $e^\prime$ is $ck$ and that of
$e$ is ``$\on{ck}{c}$''.

\begin{figure}[t]
\vspace*{-3mm}
{\hspace*{5mm}\scriptsize\begin{Verbatim}[numbers=left,xleftmargin=5mm,commandchars=\\\{\}]
node expert1(x\color{blue},train,dy\color{black})=(y\color{blue},dx\color{black})
  y = f(x,p);
  \color{blue}dxp = grad_f(x when train,p when train,dy);
  \color{blue}dx = fst(dxp); dp = snd(dxp);\color{black}
  p = init_p fby (p \color{red}+ (merge train dp 0)\color{black});
    
node gmoe(x\color{blue},train,dy\color{black})=(y\color{blue},dx\color{black})
  w1,w2\color{blue},dxg\color{black} = inst gating(x\color{blue},train,dw1,dw2\color{black});
  c1 = (w1 = 0); c2 = (w2 = 0);
  y1\color{blue},dx1\color{black} = inst expert1(x when c1\color{blue},
                        \color{blue}train when c1,dy1\color{black});
  y2\color{blue},dx2\color{black} = inst expert2(x when c2\color{blue},
                        \color{blue}train when c2,dy2\color{black});
  y = (merge c1 y1 0)*w1 + (merge c2 y2 0)*w2;
  \color{blue}dy1 = (dy*(w1 when train)) when (c1 when train);
  \color{blue}dy2 = (dy*(w2 when train)) when (c2 when train);
  \color{blue}dw1 = ((merge c1 y1 0) when train)*dy;
  \color{blue}dw2 = ((merge c2 y2 0) when train)*dy;
  \color{blue}dx  = dxg + (merge (c1 when train) dx1 0)
  \color{blue}          + (merge (c2 when train) dx2 0);
\end{Verbatim}
}
\vspace*{-2mm}
\caption{Gated mixture of experts (of \figref{gmoe}, in black), with
  gradient backpropagation (in blue) and parameter update (in red). In
  \texttt{expert1}, \texttt{grad\_f} is the gradient of function
  \texttt{f} (itself a function), and \texttt{init\_p} is the initial
  value of the parameter \texttt{p} (a
  constant).\label{fig:example1train}}
\vspace*{-3mm}
\end{figure}



\subsubsection{Hindley-Milner type system\label{sec:HM}}
To enforce clock constraints, the clock calculus
builds a {\em context} $\Gamma$ storing mutually consistent clock
assignments for all program variables and all constants (the
same constant appearing twice in a node appears twice in the
context). The typing {\em rules} of \figref{clock-rules} (in black)
then allow the {\em derivation} of {\em judgments} of two types:
\begin{itemize}
\item $\Gamma\vdash e:ck$, read as ``in the context $\Gamma$,
  expression $e$ has clock $ck$.''
\item $\Gamma\vdash f$, read as ``in the context $\Gamma$,
  program fragment $f$ is {\em well-formed},'' where $f$ can be an
  expression, list of identifiers, equation, set of equations, or node.
\end{itemize}
The construction of $\Gamma$ typically starts with a {\em base
  context} where each program variable is assigned a distinct clock
variable.
%
Then, the clock constraints of the program are progressively
considered, resulting in changes to the context through a process of
{\em unification}, formally presented in \secref{substitutions}. For
instance, assume that, in the current context $\Gamma$, the types of
variables $x$ and $y$ are respectively ``$\on{\var{u}}{c}$'' and
``$\on{\on{\var{v}}{b}}{c}$'', and assume that $x$ and $y$ are used in
the expression ``$x+y$'', so that their clocks must be equal. To
enforce this equality, the unification process will determine that
clock variable $\var{u}$ can be {\em substituted} everywhere with
``$\on{\var{v}}{b}$''. Applying this substitution to $\Gamma$
produces a new context $\Gamma^\prime$ where the clocks of both $x$
and $y$ have become ``$\on{\on{\var{v}}{b}}{c}$'', and $\var{u}$ is no
longer used.

A unification failure means that the constraint system does not have a
solution, so the program must be rejected. For instance,
clocks $\var{u}$ and ``$\on{\var{u}}{c}$'' cannot be unified, as this
would impose a constraint on the {\em value} of $c$.

\subsubsection{Limitations (motivating example)\label{sec:motivation}}
We exemplify the limitations of \cclus\ with the dataflow program of
\figref{example1train}, which provides a joint inference and training
implementation of the GMoE layer of \figref{gmoe}. It provides a
simple version of \verb|expert1|, consisting of a single parameter and
one pure function. Notice that the {\em source-to-source
  transformation} of \figref{gmoe} into \figref{example1train}
maintains the forward pass code (in black) unchanged. At each
execution cycle, the choice between pure inference and training
is determined by the Boolean input \verb|train| with clock
``\Tick''. When true, backpropagation is performed: for every forward
(inference) value \verb|v|, an {\em error} value \verb|dv| is computed
in reverse order w.r.t.\ the forward pass. This is reflected at the
level of node interfaces: the input \verb|x| of node \verb|gmoe| has
the backpropagated counterpart output \verb|dx|, whose clock is
``\verb|. on train|''. Parameter update is performed at every training
cycle on a single sample (no batch).

%

The transformation is structural: equations on lines 14--17 are the
backpropagation counterpart of line 13, and the equations on lines
18--19 are the counterpart of the broadcast of \verb|x| towards lines
7, 9, and 11.
This dataflow program provides a unified and compact description of
the inference and training semantics of the model.
\emph{Unfortunately, it is rejected by Lustre's clock calculus}.
Indeed, the instantiation of node \verb|expert1| in \verb|gmoe|
transforms the equation in line 3 into:

{\scriptsize\begin{verbatim}
  dxp = grad_f(x when c1 when (train when c1),
               p when (train when c1), dy1);
\end{verbatim}}

\noindent
In this function call, the first argument has clock
$ck_1$=``{\small$\on{\on{\Tick}{\mathtt{c1}}}{(\when{\mathtt{train}}{\mathtt{c1}})}$}''.
However, the definition of \texttt{dy1} in line 15
first subsamples on $\mathtt{train}$, so the clock of
$\mathtt{dy1}$ is
$ck_2$=``{\small$\on{\on{\Tick}{\mathtt{train}}}{(\when{\mathtt{c1}}{\mathtt{train}})}$}''.
As clocks of arguments to one function call, $ck_1$ and $ck_2$ must be
equal, yet this cannot be the case under \cclus.


\section{Related work\label{sec:related}}

Much of the related ML literature was introduced in the previous sections.
We focus here on the most closely related dataflow languages and clock calculi.

Dataflow reactive programming allows the natural representation of
large classes of ML applications and usage scenarios \cite{hugo}.  In
particular, the large variety of ML programming constructs can be
expressed using the four primitives of Lustre.  Resource allocation
can be managed and optimized, from the specification of target
architectures and non-functional requirements \cite{multiproc-part} to
static scheduling via solvers \cite{10753468} or heuristics
\cite{endochrony, from-simu}. This includes scenarios with complex
data-dependent control, conditional execution, and pipelining.


In the Lustre \cite{lustreRTSS} family of dataflow synchronous
languages \cite{CaP96,from-simu,endochrony}, activation analysis is
typically organized as a type system, with emphasis placed on
explainability and scalability rather than maximal expressiveness
\cite{CoP03,Bal08}.  One of the most expressive extensions, known as \emph{n-synchrony} \cite{Coh06}, allows for
partial desynchronization.
Its ability to concretize and flatten periodic clock
expressions---e.g.\ \texttt{(001) on (01)} is the same as
\texttt{(000001)}---facilitates the sort of commutations motivating
this paper; yet it is restricted to periodic control.

The Signal language family \cite{BENVENISTE1991103} diverges from this
trend: its activation analysis is not based on types but on a specialized
solver \cite{10.1145/207110.207134}. Extensions include an SMT-solver
approach for numeric clock equations \cite{10.1145/1967677.1967688},
and a SAT-based approach \cite{emsoft09} to capture Boolean control as
closely as possible.
With \ccnew\ we maintain the low complexity and explainability of a type system with sufficient expressiveness for ML training, real-time resource allocation \cite{emsoft09} and parallel execution \cite{girault}.



\section{The new clock calculus \ccnew\label{sec:newclocks1}}

To allow a natural and modular representation of ML training and
resource allocation we extend the Lustre clock calculus \cclus\ to
make it capable of determining that changes in the sampling order do
not change the activation condition.  The core calculus is still a
Hindley-Milner polymorphic type system, but we adapt unification to
include a semantic notion of clock equivalence that subsumes the
syntactic clock equality of Lustre. The result is a form of
\emph{equational unification}, or \emph{E-unification} \cite[Section
  3]{Ba01}.

This extension requires non-trivial changes to both clock syntax and
to the underlying type inference engine. While equational unification
is undecidable in the general case \cite[p.487]{Ba01}, our extension
of the clock syntax ensures that type checking and inference remain
tractable, even though they are extended with a {\em conversion rule}
that allows switching any clock with an equivalent one.


\subsection{Type system}
The new clock calculus must establish that expressions
$e_1=\text{``\when{x}{\when{a}{(\when{b}{a})}}''}$ and
$e_2=\text{``\when{x}{\when{b}{(\when{a}{b})}}''}$ have the same clock
whenever variables $x$, $a$, and $b$ have the same clock.
To avoid resorting to aliasing analysis, we
extend the syntax of subsampling expressions (non-terminal $\nt{sse}$
of \figref{lustre}) to allow the use of the $\code{when}$ operator.
The new definition of non-terminal $\nt{sse}$ is:
\[
\textcolor{red}{\nt{sse} \gramdef \nt{id} \mid \cnot{(\nt{sse})} \mid \when{\nt{sse}}{\nt{sse}}}
\]
If ``\Tick'' is the clock of $x$, $a$, and $b$, then 
$ck_1=$``\on{\Tick}{\on{a}{(\when{b}{a})}}'' is a clock of $e_1$.

The clocks generated by this grammar are equipped with the equivalence
relation $\sim$, defined in \figref{clock-rules}. This relation
captures the intuition that two clocks $ck_1$ and $ck_2$ are
equivalent if they have the same base clock
($\base{ck_1}=\base{ck_2}$) and are subsampled on the same literals,
regardless of their order ($\subs{ck_1}=\subs{ck_2}$).
For instance, if $ck_2=$``\on{\Tick}{\on{b}{(\when{a}{b})}}''
then $ck_2$ is a clock of $e_2$, $\subs{ck_1}=\subs{ck_2}=\{a,b\}$, 
and $ck_1\sim ck_2$.



In every context $\Gamma$, a new equivalence relation $\equiv_\Gamma$
is defined as the restriction of $\sim$ to well-formed
clocks. Whenever confusion is not possible, as in
\figref{clock-rules}, we abbreviate its notation to $\equiv$.
We only manipulate contexts where the clock assignments are mutually
consistent:
\begin{definition}
	A context $\Gamma$ is \emph{consistent} if for all $(x, ck)\in
	\Gamma$ we have $\Gamma\vdash ck$.
\end{definition}
\noindent This definition relies on that of well-formedness, which is
mutually recursive with that of subsampling expressions.  The
intuition is simple: {\em there exists a hierarchy among clocks and
  variables of a node, given by the prefix order on clocks}. When a
variable $v$ is used in the definition of a clock $ck$, the clock of
$v$ must be a prefix of $ck$. For instance, the clock
$ck=$``\on{\var{a}}{(\when{x}{y})}'' cannot be well-formed, regardless
of the context, because its only prefix is \var{a}, and placing
variables $x$ and $y$ on clock \var{a} places subsampling expression
``\when{x}{y}'' on clock ``\on{\var{a}}{y}'', so the left and right
arguments of \code{on} in $ck$ have different clocks.

\ifextended
Under a consistent $\Gamma$, we can verify that typing judgments
always produce well-formed clocks. This property also justifies the
structure of the rule \textsc{(On)}.
\begin{ext-lemma}[Well-formed soundness]\label{lemma:typing-implies-wf}
	If $\Gamma$ is a consistent context and $\Gamma\vdash e:ck$,
	then $\Gamma\vdash ck$.
\end{ext-lemma}
\begin{proof}
	Immediate consequence of \lemmaref{inversion}, as it implies
	$\Gamma\vdash ck\equiv ck'$ for some $ck'$, which, by definition
	of $\equiv$, implies $\Gamma\vdash ck$.
\end{proof}

Whenever the last (bottom) rule in a deduction tree is
(\textsc{Conv}), we can remove it to build an equivalent clocking
judgment.
\begin{ext-lemma}[Inversion]\label{lemma:inversion}
	If $\Gamma$ is consistent and $\Gamma\vdash e:ck$, then
	there exists $ck'$ such that $\Gamma\vdash ck\equiv ck'$,
	$\Gamma\vdash e:ck'$, and the last rule in the derivation of
	$\Gamma\vdash e:ck'$ is not (\textsc{Conv}).
\end{ext-lemma}
\begin{proof}
	Suppose $\Gamma$ is a consistent context and we have $\Gamma\vdash e:ck$. We prove the result by structural induction on $e$, each case being shown by induction on the derivation.
	
	For the base case, $e$ is either an identifier or a constant. Suppose it is an identifier $x$ (the constant case works in a similar fashion). Then either $\Gamma\vdash x:ck$ is shown directly by the rule (\textsc{Ident}), in which case $ck$ is the clock of $x$ in the consistent context $\Gamma$, so $\Gamma\vdash ck$, and so by reflexivity of $\equiv$ we have that $ck$ is the witness we are looking for. Or $\Gamma\vdash x:ck$ ends its derivation with the rule (\textsc{Conv}), in which case the witness $ck'$ is given by the premisses of the rule. We can then use the induction hypothesis on $\Gamma\vdash x:ck'$ and conclude by transitivity of $\equiv$.
	
	For the inductive step, the proofs are
	similar, and we consider the case of $e=\when{e'}{sse}$ as
	representative.
	We need to prove that if $\Gamma\vdash \when{e}{sse}:ck$, then there is $ck'$ such that $\Gamma\vdash ck\equiv ck'$ and $\Gamma\vdash
	\when{e}{sse}:\on{ck'}{sse}$. We proceed by induction on the
	derivation of $\Gamma\vdash e:ck$. If the last rule used is
	\textsc{(When)}, then $ck$ is directly of the form $\on{ck'}{sse}$ and
	$\Gamma\vdash sse:ck'$. By \textsc{(On)}, the latter implies that
	$\Gamma\vdash ck'$, and thus by reflexivity of $\equiv$ we can
	conclude. Otherwise the last rule used is \textsc{(Conv)}, and so
	there is $ck'$ such that $\Gamma\vdash ck\equiv ck'$ and
	$\Gamma\vdash e:ck'$.  By induction hypothesis on the latter, there
	exists $ck''$ such that $\Gamma\vdash e:\on{ck''}{sse}$ and
	$\Gamma\vdash ck'\equiv \on{ck''}{sse}$. We conclude by transitivity of $\equiv$.
\end{proof}
\fi

\ifextended
As a corollary, we obtain that typing is unique up to equivalence,
showing that the (\textsc{Conv}) rule does not introduce
inconsistencies into the type system.
\fi

\begin{theorem}[Uniqueness of typing]
	If $\Gamma$ is consistent, $\Gamma\vdash e:ck_1$, and $\Gamma\vdash e:ck_2$,
	then $\Gamma\vdash ck_1\equiv ck_2$.
\end{theorem}
\ifextended
\begin{proof}
	Structural induction over $e$.  For the inductive step, we use
	\lemmaref{inversion}, the induction hypothesis and the transitivity
	of $\equiv$.  For identifiers and constants, the context provides
	the common clock.
\end{proof}
\fi

\newlength{\myboxwidth}
\setlength{\myboxwidth}{\dimexpr\textwidth-2\fboxsep-2\fboxrule\relax}

\begin{figure*}[!t]
	\noindent\fbox{
		\begin{minipage}{\myboxwidth}
			\small
			\vspace{0.5em}
			\begin{center}\textbf{Expressions}\end{center}
			\begin{mathpar}
				\inferrule{(k, ck)\in \Gamma}{\Gamma\vdash k : ck}
				\ (\textsc{Const})
				
				\inferrule{(x, ck)\in \Gamma}{\Gamma\vdash x:ck}
				\ (\textsc{Ident})
				
				\inferrule{\Gamma\vdash e_1:ck \\ \ldots \\ \Gamma\vdash e_n:ck}{\Gamma\vdash \nt{op}(e_1, \ldots, e_n):ck}
				\ (\textsc{Op})
				
				\inferrule{\Gamma\vdash k:ck \\ \Gamma\vdash e:ck}{\Gamma\vdash \fby{k}{e} : ck}
				\ (\textsc{Fby})
				
				\inferrule{\Gamma\vdash e:ck \\ \Gamma\vdash sse:ck}{\Gamma\vdash \when{e}{sse} : \on{ck}{sse}}
				\ (\textsc{When})
				
				\inferrule{
					\Gamma\vdash sse:ck \\ 
					\Gamma\vdash e_1:\on{ck}{sse} \\
					\Gamma\vdash e_2:\on{ck}{(\cnot{sse})}
				}{
					\Gamma\vdash\merge{sse}{e_1}{e_2}:ck
				}\ (\textsc{Merge})
				
			\end{mathpar}
			
			\vspace*{-3mm}
			\begin{center}\textbf{Equations and nodes}\end{center}
			\begin{mathpar}
				\inferrule{
					\Gamma\vdash x:ck \\ \Gamma\vdash e:ck
				}{
					\Gamma\vdash \ceq{x[\annot ck]}{e}
				}\ (\textsc{Eq})
				
				\inferrule{
					\Gamma\vdash eq_1 \\ \ldots \\ \Gamma\vdash eq_n
				}{
					\Gamma\vdash eq_1 \ldots eq_n
				}\ (\textsc{Eqs})
				
				\inferrule{
					\Gamma\vdash x_1:ck_1 \\ \ldots \\ \Gamma\vdash x_n:ck_n
				}{
					\Gamma\vdash x_1[\annot ck_1],\ldots,x_n[\annot ck_n]
				}\ (\textsc{List})
				
				\inferrule{
					\Gamma\vdash \ell_{in} \\ \Gamma\vdash \ell_{out} \\ \Gamma\vdash eqs
				}{
					\Gamma\vdash \node{\nt{id}}{\ell_{in}}{\ell_{out}}\ eqs
				}\ (\textsc{Node})
				
				\inferrule{
					\Gamma\vdash node \\ \Gamma\vdash prog
				}{
					\Gamma\vdash node\ prog
				}\ (\textsc{Prog})
			\end{mathpar}
			
			\vspace*{-3mm}
			\begin{center}\textbf{Well-formedness and \color{blue}equivalence}\end{center}
			\begin{mathpar}
                                \color{blue}
				\inferrule{
					\Gamma\vdash e:ck \\ \Gamma\vdash ck\equiv ck'
				}{
					\Gamma\vdash e: ck'
				}\ (\textsc{Conv})

                                \color{black}
				\inferrule{}{
					\Gamma\vdash \var{x}
				}\ (\textsc{Var})

				\inferrule{
					\Gamma\vdash sse:ck
				}{
					\Gamma\vdash \on{ck}{sse}
				}\ (\textsc{On})
                                
				\inferrule{}{
					\Gamma\vdash\Tick
				}\ (\textsc{Tick})

                                \color{blue}
				\inferrule{
					\base{ck}=\base{ck'} \\ \subs{ck} = \subs{ck'}
				}{
					ck\sim ck'
				}\ (\textsc{Equiv})
				
				\inferrule{
					\Gamma\vdash ck \\ \Gamma\vdash ck' \\ ck\sim ck'
				}{
					\Gamma\vdash ck\equiv ck'
				}\ (\textsc{Equiv-$\Gamma$})

			\end{mathpar}
			
			\vspace*{-3mm}
			\begin{center}\color{blue}\textbf{Helper functions}\end{center}
			\vspace{-8mm}
			\begin{center}
				\begin{minipage}[t]{0.43\textwidth}
					\raggedright\color{blue}
					\begin{align*}
						&\base{\Tick} = \Tick \;\;\;\;\;\;\;\;\;\base{\var{x}} = \var{x}\\
						&\base{\on{ck}{sse}} = \base{ck} \\
						\\
						&\subs{\Tick} = \varnothing \;\;\;\;\;\;\;\;\;\subs{\var{x}} = \varnothing  \\
						&\subs{\on{ck}{sse}} = \subs{ck}\cup\{\ebase{sse}\}
					\end{align*}
				\end{minipage}
				\begin{minipage}[t]{0.52\textwidth}
					\color{blue}\begin{align*}
						& \ebase{x} = x \\
						& \ebase{\when{sse_1}{sse_2}} = \ebase{sse_1} \\
						& \ebase{\cnot{(sse)}} = {\begin{cases}
								x & \text{if }\ebase{sse}=\cnot{x}\\
								\cnot{x} & \text{if }\ebase{sse}=x
						\end{cases}}
					\end{align*}
				\end{minipage}
			\end{center}
		\end{minipage}
	}
	\centering
	\caption{Clock inference rules of \cclus\ (in black) and \ccnew\ extensions (in blue)}
	\label{fig:clock-rules}
        \vspace*{-5mm}
\end{figure*}

\subsection{Substitutions\label{sec:substitutions}}
The fundamental tool of a Hindley-Milner type system is the {\em
  substitution}:
\begin{definition}[Substitution]
	A \emph{substitution} is a function from clock
        variables to clocks. Every substitution $\sigma$ is naturally
        extended to all clocks by setting $\sigma(\Tick) = \Tick$ and
	$\sigma(\on{ck}{sse}) = \on{\sigma(ck)}{sse}$.
\end{definition}
\noindent The domain of a substitution $\sigma$, denoted $dom(\sigma)$, is the
set of clock variables $\var{v}$ such that
$\sigma(\var{v})\not=\var{v}$. The \emph{identity substitution} with
$dom(\sigma)=\varnothing$ is denoted $\id$.
\ifextended
For simplicity, we restrict the domain of our substitutions to the variables present in the context. This is without loss of generality.
\fi

During clock inference, enforcing the constraint that $ck_1$ and
$ck_2$ are equal (formally denoted $ck_1\doteq ck_2$) amounts to
finding a substitution $\sigma$ such that $\sigma(ck_1)$ and
$\sigma(ck_2)$ are equivalent. For instance, the substitution
$\sigma=[\var{b}\mapsto\on{\var{a}}{x}]$ solves the constraint
$\on{\var{a}}{x} \doteq \var{b}$. A substitution $\sigma$ induces a
context transformation
$\Gamma\mapsto\sigma(\Gamma)=\{(x,\sigma(ck))\mid (x,ck)\in\Gamma\}$.
As functions on contexts, substitutions can be composed.

\ifextended
\begin{ext-definition}[Unifier, Unifiable clocks]
	Let $\Gamma$ be a consistent context and $C=\{ck_1\doteq
	ck^\prime_1,\ldots,ck_n\doteq ck^\prime_n\}$ a set of clock
	constraints. A substitution $\sigma$ is a \emph{unifier} for $C$ under
	$\Gamma$ if $\sigma(\Gamma)$ is consistent and
	$\forall i:\sigma(\Gamma)\vdash \sigma(ck_i)\equiv
	\sigma(ck^\prime_i)$.
	
	If a unifier exists for $C=\{ck_1\doteq ck_2\}$ then we say that
	$ck_1$ and $ck_2$ are \emph{unifiable}.
\end{ext-definition}
As context transformations, substitutions can be composed, using
the classical notation $\sigma \circ \tau$. Composition is associative
with $id$ as neutral element.
\fi

\begin{definition}[Valid substitution]
	A substitution $\sigma$ is \emph{valid} in a context $\Gamma$ if
	$\sigma(\Gamma)\vdash \sigma(\var{x})$ for all $\var{x}\in dom(\sigma)$.
\end{definition}
\begin{prop}[Consistency preservation]\label{prop:consistency-preservation}
	Let $\Gamma$ be a consistent context.
	A substitution $\sigma$ is valid in context $\Gamma$ if and only if $\sigma(\Gamma)$ is consistent.
        Furthermore, when this happens, if $\Gamma\vdash e:ck$ then $\sigma(\Gamma)\vdash e:\sigma(ck)$,
        and if $\Gamma\vdash ck\equiv ck'$ then $\sigma(\Gamma)\vdash \sigma(ck)\equiv\sigma(ck')$.
\end{prop}
\ifextended
\begin{proof}
	For the direct implication.
	Let $ck\in ran(\sigma(\Gamma))$.
	By definition, there is $ck'\in ran(\Gamma)$ such that
	$ck=\sigma(ck')$.
	By consistency of $\Gamma$, we have $\Gamma\vdash ck'$, and by
	item (a) of \lemmaref{wf-preservation} (provided next), we
	obtain $\sigma(\Gamma)\vdash \sigma(ck')$, yielding the
	desired result $\sigma(\Gamma)\vdash ck$.
	
	For the converse.
	Since the domain of $\sigma$ is restricted to the variables present in $\Gamma$, there a two possibilities.
	Either $dom(\sigma)=\varnothing$ and so $\sigma=\id$ which is vacuously valid over $\Gamma$.
	Or $dom(\sigma)\neq\varnothing$, in which case, consider any
	$\var{x}\in dom(\sigma)$. Since $\var{x}$ appears in $\Gamma$,
	there is an identifier $x$ such that $(x, ck)\in\Gamma$ where
	$ck=\on{\on{\on{\var{x}}{e_1}}{\ldots}}{e_n}$
	and so, by consistency of $\sigma(\Gamma)$ we have
	$\sigma(\Gamma)\vdash \on{\on{\on{\sigma(\var{x})}{e_1}}{\ldots}}{e_n}$.
	By a straightforward induction on the rule \textsc{(On)}, the
	derivation of this judgment must contain a derivation for
	$\sigma(\Gamma)\vdash \sigma(\var{x})$.
\end{proof}
\fi
\noindent In other words, once a constraint is satisfied, it remains satisfied throughout the inference process.

\ifextended
\begin{ext-lemma}[Substitution preservation]\label{lemma:wf-preservation}\label{lemma:subst-preservation}
	If $\sigma$ is a valid substitution over a consistent context $\Gamma$, then
	\begin{enumerate}
		\item if $\Gamma\vdash ck$ then $\sigma(\Gamma)\vdash\sigma(ck)$
		\item if $\Gamma\vdash sse:ck$ then $\sigma(\Gamma)\vdash sse:\sigma(ck)$
	\end{enumerate}
\end{ext-lemma}
\begin{proof}
	We prove $(1)$ and $(2)$ by a mutual induction on derivation trees.
	
	\textbf{Cases for (1)}.
	Suppose $\Gamma\vdash b$ for $b\in\{\Tick\}\cup\{\var{x}:\var{x}\notin dom(\sigma)\}$ then $\sigma(b)=b$, and by \textsc{(Tick)} or \textsc{(Var)} we have $\sigma(\Gamma)\vdash b$. If $\Gamma\vdash \var{x}$ where $\var{x}\in dom(\sigma)$, then $\sigma(\Gamma)\vdash \sigma(\var{x})$ by validity of $\sigma$.
	
	Suppose $\Gamma\vdash \on{ck}{sse}$. This must derives from $\Gamma\vdash sse:ck$. By induction hypothesis for (b), we obtain that $\sigma(\Gamma)\vdash sse:\sigma(ck)$. Thus, by \textsc{(On)}, we obtain that $\sigma(\Gamma)\vdash \on{\sigma(ck)}{sse}$ which is exactly $\sigma(\Gamma)\vdash\sigma(\on{ck}{sse})$.
	
	\textbf{Cases for (2)}.
	Suppose $\Gamma\vdash x:ck$ such that it derives from $(x, ck)\in\Gamma$. Then $(x, \sigma(ck))\in\sigma(\Gamma)$ by definition, and thus, by applying \textsc{(Ident)}, we obtain $\sigma(\Gamma)\vdash x:\sigma(ck)$.
	
	Suppose $\Gamma\vdash \cnot{sse}:ck$ derives from $\Gamma\vdash sse:ck$. Then, by induction hypothesis we have $\sigma(\Gamma)\vdash sse:\sigma(ck)$, and we can conclude by applying \textsc{(Not)}.
	
	Suppose $\Gamma\vdash\when{sse_1}{sse_2}:ck$ derives from the rule \textsc{(When)}. Then $ck$ must be of the form $\on{ck'}{sse_2}$, and we have $\Gamma\vdash sse_1:ck'$ and $\Gamma\vdash sse_2:ck'$. By aplying induction hypothesis on both we obtain $\sigma(\Gamma)\vdash sse_1:\sigma(ck')$ and $\sigma(\Gamma)\vdash sse_2:\sigma(ck')$. Thus, by applying \textsc{(When)}, we have $\sigma(\Gamma)\vdash\when{sse_1}{sse_2}:\on{\sigma(ck')}{sse_2}$ which is exactly $\sigma(\Gamma)\vdash\when{sse_1}{sse_2}:\sigma(\on{ck'}{sse_2})$.
	
	Finally, suppose $\Gamma\vdash sse:ck$ derives from \textsc{(Conv)}. Then there is $ck'$ such that $\Gamma\vdash sse:ck'$ and $\Gamma\vdash ck\equiv ck'$. Where the latter must derive from the \textsc{(Equiv)} rule, and so we have $\Gamma\vdash ck$, $\Gamma\vdash ck'$ and $ck\sim ck'$. By induction hypothesis for (a) on these shorter derivations we obtain $\sigma(\Gamma)\vdash sse:\sigma(ck')$, $\sigma(\Gamma)\vdash \sigma(ck)$ and $\sigma(\Gamma)\vdash \sigma(ck')$. Moreover since $ck\sim ck'$ we have $\sigma(ck)\sim \sigma(ck')$, and thus, by applying \textsc{(Equiv)} we have $\sigma(\Gamma)\vdash \sigma(ck)\equiv\sigma(ck')$. Finally, we apply \textsc{(Conv)} to obtain $\sigma(\Gamma)\vdash sse:\sigma(ck)$, as desired.
\end{proof}

The last case of the proof states that clock equivalence is stable under valid substitutions. This ensures that once a constraint is satisfied, it remains satisfied throughout the inference process.
Formally, if $\Gamma$ is a consistent context and $\sigma$ is valid over $\Gamma$, then $\sigma$ preserves equivalence, i.e. if $\Gamma\vdash ck\equiv ck'$ then $\sigma(\Gamma)\vdash \sigma(ck)\equiv\sigma(ck')$.

Another consequence of the lemma is that valid substitutions can be composed.
\begin{ext-lemma}\label{lemma:valid-subst-composition}
	Let $\Gamma$ be a consistent context. If $\sigma$ is valid under $\Gamma$ and $\tau$ is valid under $\sigma(\Gamma)$,
	then $\tau\circ\sigma$ is valid under $\Gamma$.
\end{ext-lemma}
\begin{proof}
	Let $\var{x}\in dom(\tau\circ\sigma)$. We show that $\tau\circ\sigma(\Gamma)\vdash \tau\circ\sigma(\var{x})$. There are two cases.
	If $\var{x}\in dom(\sigma)$ then $\sigma(\Gamma)\vdash \sigma(\var{x})$ by validity. By \lemmaref{subst-preservation}, since $\tau$ is valid under $\sigma(\Gamma)$ we obtain the desired result.
	If $\var{x}\notin dom(\sigma)$, then $\sigma(\var{x})=\var{x}$ and so $\tau\circ\sigma(\var{x})=\tau(\var{x})$.
	By \textsc{(Var)} we have that $\sigma(\Gamma)\vdash \var{x}$, and so by \lemmaref{subst-preservation}, since $\tau$ is valid under $\sigma(\Gamma)$, we obtain $\tau\circ\sigma(\Gamma)\vdash \tau(\var{x})$.
\end{proof}
\fi

\begin{definition}[Equivalence of substitutions]
	We say that two substitutions $\sigma$ and $\tau$ are equivalent under
	context $\Gamma$, denoted $\Gamma\vdash \sigma\equiv\tau$, if
	both $\sigma$ and $\tau$ are valid under $\Gamma$, $dom(\sigma) = dom(\tau)$,
        and $\forall \var{x}\in dom(\sigma), \sigma(\var{x})\sim \tau(\var{x})$.
\end{definition}
\begin{definition}[Substitution preorder $\lesssim$]
	Substitution $\sigma$ is \emph{more general} than $\tau$ under
        $\Gamma$, denoted $\Gamma\vdash\sigma\lesssim\tau$, if there
        exists a substitution $\gamma$, valid under $\sigma(\Gamma)$,
        such that $\Gamma\vdash \tau\equiv\gamma\circ\sigma$.
\end{definition}
\begin{definition}[Most general unifier]
  Given a consistent context $\Gamma$ and a set of clock constraints
  $C=\{ck_1\doteq ck_1^\prime,\ldots,ck_n\doteq ck_n^\prime\}$, we say that
  substitution $\sigma$ is the most general unifier of $C$ under $\Gamma$
  if $\forall i:\sigma(\Gamma)\vdash \sigma(ck_i)\equiv\sigma(ck^\prime_i)$
  and if $\sigma$ is minimal with this property under $\lesssim$.
\end{definition}


\subsection{Clock inference algorithm}\label{sec:unification}
Clock inference is realized using a variant of Algorithm~W
\cite{DamasMilner82}. All expressions, equations, clocks and clock
annotations of the program are traversed exactly once, in
sequence.
Starting from the base context $\Gamma_0$, each step $n\geq 1$
of the traversal
incrementally builds a new valid substitution $\tau_n$ and a
new context $\Gamma_n=\tau_n(\Gamma_{n-1})$. The context $\Gamma_n$
ensures the satisfaction of all constraints associated with the
expressions, equations, clocks and clock annotations already
traversed.

The traversal respects the hierarchy of the program: a subexpression
is traversed before the expression, equation, or clock containing it;
a clock prefix (which is a clock itself) is traversed before the full
clock.
For example, consider the expression
$e=$ ``$\when{x}{(\when{a}{b})}+\when{y}{(\when{b}{a})}$''
which is representative of the problems posed by our motivational
example. Its subexpressions are $e_1=$``$\when{a}{b}$'', $e_2=$``$\when{b}{a}$'',
$e_3$ the left operand of +, and $e_4$ its right operand. A valid
traversal order is $e_1, e_2, e_3, e_4, e$, and the initial context
is
$\Gamma_0=\{(a,\var{a}),(b,\var{b}),(x,\var{x}),(y,\var{y})\}$.  The
clock constraint associated with $e_1$ is that the clocks of $a$ and
$b$ are equal, i.e. $\var{a}\doteq\var{b}$. The most general unifier for
$\var{a}\doteq\var{b}$ under $\Gamma_0$ is
$\tau_1=[\var{b}\mapsto\var{a}]$, resulting in
$\Gamma_1=\{(a,\var{a}),(b,\var{a}),(x,\var{x}),(y,\var{y})\}$.  Since
$e_2$ has the same clock constraint as $e_1$, we have $\tau_2=\id$ and $\Gamma_2=\Gamma_1$.  As
$\Gamma_2\vdash e_1:\on{\var{a}}{b}$, the clock constraint associated
with $e_3$ is $\var{x}\doteq\on{\var{a}}{b}$. The most general unifier
is $\tau_3=[\var{x}\mapsto\on{\var{a}}{b}]$, and so
$\Gamma_3=\{(a,\var{a}),(b,\var{a}),(x,\on{\var{a}}{b}),(y,\var{y})\}$.
Similarly, $\tau_4=[\var{y}\mapsto\on{\var{a}}{a}]$, and
$\Gamma_4=\{(a,\var{a}),(b,\var{a}),(x,\on{\var{a}}{b}),(y,\on{\var{a}}{a})\}$.
Upon traversal of $e$ we have
$\Gamma_4\vdash e_3:\on{\on{\var{a}}{b}}{(\when{a}{b})}$ and
$\Gamma_4\vdash e_4:\on{\on{\var{a}}{a}}{(\when{b}{a})}$. Since these
two clocks are equivalent under context $\Gamma_4$, the constraint is satisfied,
and so $\tau_5=\id$ and $\Gamma_5=\Gamma_4$.

The hierarchical traversal greatly simplifies clock inference. When
building the most general unifier $\tau_i$ for the clock constraint
$ck_i\doteq ck_i^\prime$ we only have to consider 4 cases:
\begin{itemize}
	\item If $ck_i\sim ck_i^\prime$, then nothing needs to be done, i.e. $\tau_i=\id$.
	\item If $ck_i=\var{x}$ and $\base{ck_i^\prime}\neq \var{x}$, then $\tau_i=[\var{x}\mapsto ck_i^\prime]$.
	\item If $ck_i^\prime=\var{x}$ and $\base{ck_i}\neq \var{x}$, then $\tau_i=[\var{x}\mapsto ck_i]$.
	\item All other cases are typing errors, as no solution exists.
\end{itemize}

\ifextended
\subsection*{Technical preliminaries}
To tackle the upcoming proofs of correctness we need
to establish a series of technical results on clocks, derivations and consistent contexts.
In this section, for a literal $\ell$, we denote by $|\ell|$ its underlying identifiers, i.e. $|\ell| = c$ if $\ell=c$ or $\ell=\cnot{c}$, for some identifier $c$.
We start with results on clocks and their subsamplings.

\begin{ext-lemma}\label{lemma:prefix-wf}
	Let $\Gamma$ be a consistent context. If $\Gamma\vdash e : \on{\on{\on{\beta}{s_1}}{\ldots}}{s_k}$, where $\beta$ is either $\Tick$ or a clock variable and $k\geq 0$ some integer, then for all $i\in\{1, \ldots, k\}$ we have $\Gamma\vdash s_i:\on{\on{\on{\beta}{s_1}}{\ldots}}{s_{i-1}}$.
\end{ext-lemma}
\begin{proof}
	Since $\Gamma$ is consistent, we can use \lemmaref{typing-implies-wf} to obtain $\Gamma\vdash \on{\on{\on{\beta}{s_1}}{\ldots}}{s_k}$.
	That derivation can only be obtained by the \textsc{(On)} rule, thus the premise $\Gamma\vdash s_k:\on{\on{\on{\beta}{s_1}}{\ldots}}{s_{k-1}}$ must hold. We can then repeat the same procedure for all $i\leq k$.
\end{proof}

\begin{ext-lemma}[Clock of subsampling]\label{lemma:ck-sub}
	Let $\Gamma$ be a consistent context. If $\Gamma\vdash sse:ck$,
	then, for $c=|ebase(sse)|$, we have $base(\Gamma(c))=base(ck)$ and
	$subs(\Gamma(c))\subseteq subs(ck)$.
\end{ext-lemma}
\begin{proof}
	By structural induction over $sse$. In each case, we apply
	\lemmaref{inversion} to recover the structure of $ck$ and conclude
	via the induction hypothesis.
\end{proof}

From here, we can formally deduce that a clock cannot be self referential.
Given a context $\Gamma$, we denote with $dom(\Gamma)$ the
set of identifiers and constants it associates clocks to. For $v\in
dom(\Gamma)$ we denote $\Gamma(v)$ the clock associated by $\Gamma$ to
$v$.
\begin{ext-lemma}[No self reference]\label{lemma:no-self-ref}
	If $\Gamma$ is a consistent context, then for every identifier $v\in dom(\Gamma)$ we have that $subs(\Gamma(v))$ contains no literal over $v$, i.e. $\forall \ell\in subs(\Gamma(v)), |\ell|\neq v$.
\end{ext-lemma}
\begin{proof}
	We write $\Gamma(v)=\on{\on{\on{ck}{s_1}}{\ldots}}{s_k}$. By contradiction, suppose there is a minimal index $j$ such that $|ebase(s_j)|=v$. Since $\Gamma$ is consistent, by \lemmaref{prefix-wf}, we have that $\Gamma\vdash s_j:\on{\on{\on{ck}{s_1}}{\ldots}}{s_{j-1}}$.
	Hence, by \lemmaref{ck-sub}, we have that
	\begin{align*}
		subs(\Gamma(v)) &\subseteq subs(\on{\on{\on{ck}{s_1}}{\ldots}}{s_{j-1}})\\
		&=\{ebase(s_1), \ldots, ebase(s_{j-1})\}
	\end{align*}
	A contradiction since the left-hand side contains $v$ as underlying literal of $ebase(s_j)$, but the right-hand side does not, by minimality of $j$.
\end{proof}


To explore and manipulate the structure of judgment derivations, we define a notion that captures
the idea of the rule \textsc{(On)}: that a well-formed clock can only be built from smaller well-formed clocks.

\begin{ext-definition}[Rank]
	The \emph{rank} of a clock $ck$ is defined as $rank(ck) = card(subs(ck))$.
\end{ext-definition}

Note that $rank(ck)=0$ if and only if $ck=\Tick$ or $ck=\var{x}$ for some $\var{x}$. And $rank$ is invariant under $\sim$, i.e. if $ck_1\sim ck_2$ then $rank(ck_1)= rank(ck_2)$.

\begin{ext-lemma}[Derivation structure]\label{lemma:derivation-structure}
	Let $\Gamma$ be a consistent context.
	\begin{enumerate}
		\item[(a)] If $\Gamma\vdash ck$. Either $rank(ck)=0$, in which case the derivation is a single rule \textsc{(Tick)} or \textsc{(Var)}. Or every \textsc{(Ident)} leaf whose premise is $(c, \Gamma(c))$ satisfies
		\begin{itemize}
			\item $\exists d\in subs(ck), c=|d|$
			\item $base(\Gamma(c)) = base(ck)$
			\item $rank(\Gamma(c)) < rank(ck)$
		\end{itemize}
		\item[(b)] If $\Gamma\vdash sse:ck$. Then, for every \textsc{(Ident)} leaf whose premise is $(c, \Gamma(c))$ such that $c\neq |ebase(sse)|$, we have
		\begin{itemize}
			\item $\exists d\in subs(ck), c=|d|$
			\item $base(\Gamma(c))=base(ck)$
			\item $rank(\Gamma(c)) < rank(ck)$
		\end{itemize}
	\end{enumerate}
\end{ext-lemma}
\begin{proof}
	First, note that a well-formedness judgment
	$\Gamma\vdash ck$ can only be concluded by \textsc{(Tick)},
	\textsc{(Var)} or \textsc{(On)}.
	The only context-reading rule occurring in these
	derivations is \textsc{(Ident)}, since subsampling expressions contain
	no constants.
	
	We prove (a) and (b) simultaneously, by strong induction on the size
	of the derivation.
	
	\medskip
	\textbf{Part (a).} Let $\mathcal{D}$ be a derivation of
	$\Gamma\vdash ck$.
	
	If $\rank{ck}=0$, then $ck$ is $\Tick$ or some clock variable, so $\mathcal{D}$ must be a single rule \textsc{(Tick)} or \textsc{(Var)}.
	
	If $\rank{ck}\geqslant 1$, then $ck$ is of the form $\on{ck'}{sse}$,
	and the last rule of $\mathcal{D}$ must be \textsc{(On)}, with premise $\Gamma\vdash sse:ck'$. Let $\mathcal{D}'$ be the derivation of that premise. Note that
	$\base{ck}=\base{ck'}$,
	$\subs{ck}=\subs{ck'}\cup\{\ebase{sse}\}$, and $\rank{ck'}\leqslant\rank{ck}$. Every \textsc{(Ident)} leaf in $\mathcal{D}$ is also a leaf in $\mathcal{D}'$. Consider such a leaf with premise $(c,\Gamma(c))$:
	\begin{itemize}
		\item either $c\neq|\ebase{sse}|$, in which case we can conclude using part (b) of the induction hypothesis applied to $\mathcal{D}'$;
		\item or $c=|\ebase{sse}|$, then \lemmaref{ck-sub} applied to the
		judgment $\Gamma\vdash sse:ck'$ yields
		$\base{\Gamma(c)}=\base{ck'}=\base{ck}$ and
		$\subs{\Gamma(c)}\subseteq\subs{ck'}\subseteq\subs{ck}$;
                let $d = ebase(sse)$: by \lemmaref{no-self-ref},
                we have that $d\notin\subs{\Gamma(c)}$;
                since $d\in\subs{ck}$, the inclusion
		$\subs{\Gamma(c)}\subseteq\subs{ck}$ is strict,
                hence $\rank{\Gamma(c)}<\rank{ck}$.
	\end{itemize}
	
	\medskip
	\textbf{Part (b).} Let $\mathcal{D}$ be a derivation of
	$\Gamma\vdash sse:ck$. We consider each case of its last rule.
	
	\emph{Case \textsc{(Ident)}.} Then $sse$ is an identifier $c=|\ebase{sse}|$, which is not subject to the claim.
	
	\emph{Case \textsc{(Not)}.} Then $sse=\cnot{(sse')}$, and its premise is $\Gamma\vdash sse':ck$. We have $|\ebase{sse}|=|\ebase{sse'}|$, so we can conclude directly by induction hypothesis.
	
	\emph{Case \textsc{(When)}.} Then $sse=\when{sse_1}{sse_2}$, with
	premises $\mathcal{D}_1$ deriving $\Gamma\vdash sse_1:ck'$ and
	$\mathcal{D}_2$ deriving $\Gamma\vdash sse_2:ck'$, and
	$ck=\on{ck'}{sse_2}$.
        Hence $\base{ck}=\base{ck'}$, \\
	$\subs{ck}=\subs{ck'}\cup\{\ebase{sse_2}\}$, \\
	$\rank{ck'}\leqslant\rank{ck}$, and $\ebase{sse}=\ebase{sse_1}$. Let
	$(c,\Gamma(c))$ be a leaf of $\mathcal{D}$ with
	$c\neq|\ebase{sse_1}|$.
	\begin{itemize}
		\item If it is a leaf of $\mathcal{D}_1$, then the induction hypothesis on $\mathcal{D}_1$ applies and yields the three statements with respect to $ck'$. By transitivity we can conclude for $ck$.
		\item If it is a leaf of $\mathcal{D}_2$ with $c\neq|\ebase{sse_2}|$, then we conclude by induction hypothesis on $\mathcal{D}_2$ and transitivity.
		\item If it is a leaf of $\mathcal{D}_2$ with $c=|\ebase{sse_2}|$, then let $d=\ebase{sse_2}\in\subs{ck}$. By applying \lemmaref{ck-sub} to
		$\Gamma\vdash sse_2:ck'$ we obtain
		$\base{\Gamma(c)}=\base{ck'}=\base{ck}$ and
		$\subs{\Gamma(c)}\subseteq\subs{ck'}\subseteq\subs{ck}$. By
		\lemmaref{no-self-ref}, $d\notin\subs{\Gamma(c)}$, and since
		$d\in\subs{ck}$ the inclusion is strict and yields $\rank{\Gamma(c)}<\rank{ck}$.
	\end{itemize}
	
	\emph{Case \textsc{(Conv)}.} The premises $\Gamma\vdash sse:ck'$, whose derivation is denoted $\mathcal{D}'$, and
	$\Gamma\vdash ck\equiv ck'$, whose derivation necessarily ends with
	\textsc{(Equiv-$\Gamma$)}, and thus carries premises $\Gamma\vdash ck$ (whose derivation is denoted $\mathcal{D}_1$), $\Gamma\vdash ck'$ (whose derivation is denoted $\mathcal{D}_1$), and $ck\sim ck'$. The leaves of $\mathcal{D}$ are those of
	$\mathcal{D}'$, $\mathcal{D}_1$ and $\mathcal{D}_2$.
	For the leaves of $\mathcal{D}'$ with $c\neq|\ebase{sse}|$,
        we can use the induction
	hypothesis to obtain the results on $ck'$,
        and then conclude by using $ck\sim ck'$.
	For the leaves of $\mathcal{D}_1$ or $\mathcal{D}_2$, we apply part (a)
	of the induction hypothesis to these strictly smaller derivations.
\end{proof}

To show that the substitutions generated during inference are valid
we will prove they are both idempotent and well-formed, as defined below.
\begin{ext-definition}
	A substitution $\sigma$ is \emph{idempotent} if and only if 
	$\sigma\circ\sigma = \sigma$ holds.
\end{ext-definition}
Intuitively, idempotent substitutions are those that apply all their
changes in a single step, meaning a second application has no further
effect. An example of a substitution that is {\em not} idempotent is
$[\var{x}\mapsto \var{y}, \var{y}\mapsto \var{z}]$. Idempotence is
guaranteed during inference by ensuring that the variables appearing
in the clocks of the substitution are not part of its domain.
\begin{ext-definition}
	A substitution $\sigma$ is \emph{well-formed} under a context
	$\Gamma$ if $\Gamma\vdash \sigma(\var{x})$ for all $\var{x}\in
	dom(\sigma)$.
\end{ext-definition}
Notice that $\sigma$ is valid under $\Gamma$ if and only if
$\sigma$ is well-formed under $\sigma(\Gamma)$.

The following lemma establishes that the substitutions generated by
our algorithm are valid.
\begin{ext-lemma}\label{lemma:wf-implies-valid}
	Let $\Gamma$ be a context.
	If $\sigma$ is idempotent and well-formed under $\Gamma$, then $\sigma$ is valid under $\Gamma$.
\end{ext-lemma}
\begin{proof}
	Let $\var{x}\in dom(\sigma)$, since $\sigma$ is well-formed under $\Gamma$, we have $\Gamma\vdash \sigma(\var{x})$.
	By \lemmaref{derivation-structure}, the leaves are all \textsc{(Ident)} whose premise is $(|d|, ck_d)\in \Gamma$
	where $d\in subs(\sigma(\var{x}))$ and $base(ck_d)=base(\sigma(\var{x}))$.
	By idempotency of $\sigma$ we have that $base(\sigma(\var{x}))\notin dom(\sigma)$, thus $\sigma(ck_d)=ck_d$.
	Therefore $(|d|, ck_d)\in \sigma(\Gamma)$, and so the derivation of $\Gamma\vdash \var{x}$ is also a derivation of $\sigma(\Gamma)\vdash \sigma(\var{x})$.
\end{proof}

Finally, we extend our apparatus for manipulating contexts and establish consistency results,
starting with the $\sim$ and $\equiv$ relations.
\begin{ext-definition}[Equivalence of contexts]
	We write $\Gamma\sim\Gamma'$ when $dom(\Gamma)=dom(\Gamma')$ and
	$\forall x\in dom(\Gamma), \Gamma(x)\sim \Gamma'(x)$.
	
	Two contexts $\Gamma$ and $\Gamma'$ are \emph{equivalent}, denoted
	$\Gamma\equiv\Gamma'$, if they are both consistent, and
	$\Gamma\sim\Gamma'$.
\end{ext-definition}

The following lemma relates equivalence of substitutions to equivalence of contexts. Recall that the domain of the substitutions we consider is restricted to the ones appearing in the context.
\begin{ext-lemma}\label{lemma:equiv-subs-consistent-context}
	Let $\Gamma$ be a consistent context.
	We have $\Gamma\vdash\sigma\equiv\tau$ if and only if $\sigma(\Gamma)\equiv\tau(\Gamma)$.
\end{ext-lemma}
\begin{proof}
	For the direct implication. Since $\sigma$ and $\tau$ are both valid under the consistent context $\Gamma$, by \lemmaref{subst-preservation}, we have that $\sigma(\Gamma)$ and $\tau(\Gamma)$ are both consistent. Then establishing $\sigma(\Gamma)\sim\tau(\Gamma)$ is straightforward, by using $\sigma\sim\tau$ on each clock in $\Gamma$.
	
	The reverse implication can be shown by contraposition, i.e. we suppose $\Gamma\nvdash\sigma\equiv\tau$, and show that $\sigma(\Gamma)\not\equiv\tau(\Gamma)$. The hypothesis breaks down into different assumptions.
	
	Either $\sigma$ (or $\tau$) is not valid under $\Gamma$, i.e. there is $\var{x}\in dom(\sigma)$ such that $\sigma(\Gamma)\nvdash\sigma(\var{x})$. So it suffices to show that $\sigma(\var{x})$ is a clock in $\sigma(\Gamma)$ to prove $\sigma(\Gamma)$ is not consistent. Since $\Gamma$ is consistent, and $dom(\sigma)\subseteq dom(\Gamma)$, there must be an identifier or constant whose clock is $\var{x}$. Indeed, $\var{x}\in dom(\Gamma)$ means $\Gamma$ contains a clock $ck$ of the form $\on{\on{\on{\var{x}}{e_1}}{\ldots}}{e_n}$ for some $n$. If $n=0$ we are done. Otherwise, and since $\Gamma\vdash ck$ then the identifier $|e_1|\in\Gamma$ must have clock $\var{x}$, and hence $\sigma(\var{x})$ is a clock that appears in $\sigma(\Gamma)$. 
	
	Or $dom(\sigma)\neq dom(\tau)$, in which case there is (WLOG) a variable $\var{x}\in dom(\sigma)\setminus dom(\tau)$. With the same reasoning as above we have that $\var{x}$ is a clock appearing in $\Gamma$ and so $\sigma(\var{x})\not\sim\tau(\var{x})$ gives us that $\sigma(\Gamma)\not\equiv\tau(\Gamma)$.
	
	Finally if there is $\var{x}\in dom(\sigma)$ such that $\sigma(\var{x})\not\sim\tau(\var{x})$, then we can conclude with the same reasoning as above.
\end{proof}

We can now establish the soundness of the equivalence relation on
contexts: equivalent contexts induce the same clocking and typing judgments.
That is, they are extensionally indistinguishable regarding the
clocking of expressions.

\begin{ext-lemma}[Cross well-formedness]\label{lemma:cross-wf}
	If $\Gamma\equiv\Gamma'$ then for all $v\in dom(\Gamma)$ we have both $\Gamma\vdash\Gamma'(v)$ and $\Gamma'\vdash\Gamma(v)$.
\end{ext-lemma}
\begin{proof}
	By symmetry, it is sufficient to prove $\forall v\in dom(\Gamma), \Gamma'\vdash\Gamma(v)$.
	
	By induction on $n$ we prove that for all $v\in dom(\Gamma)$ such that $rank(\Gamma(v))\leq n$ we have $\Gamma'\vdash \Gamma(v)$.
	
	For $n=0$, $\Gamma(v)$ is either $\Tick$ or some variable, so a single rule \textsc{(Tick)} or \textsc{(Var)} yields the wanted derivation in $\Gamma'$.
	
	For $n\rightarrow n+1$, let $v\in dom(\Gamma)$ be such that $rank(\Gamma(v))=n+1$. By consistency of $\Gamma$ we have $\Gamma\vdash \Gamma(v)$. We fix a derivation $\mathcal{D}$ of that statement. All the leaves in $\mathcal{D}$ are either \textsc{(Tick)}, \textsc{(Var)} or \textsc{(Ident)} rules. By \lemmaref{derivation-structure}, all the \textsc{(Ident)} rules (with premise $(c, \Gamma(c))$) verify that $rank(\Gamma(c)) < rank(\Gamma(v)) = n+1$.
	
	Each such \textsc{(Ident)} rule is replaced by a derivation of $\Gamma'\vdash c:\Gamma(c)$. The latter is obtained by using the \textsc{(Conv)} rule. Its first premise is $\Gamma'\vdash c:\Gamma'(c)$, which holds since $\Gamma'$ has the same domain as $\Gamma$. The other premise is $\Gamma'\vdash \Gamma(c)\equiv \Gamma'(c)$. It relies on three more premises, namely $\Gamma'\vdash \Gamma'(c)$ which holds by consistency of $\Gamma'$, $\Gamma'\vdash \Gamma(c)$ which holds by induction hypothesis since $rank(\Gamma(c))\leq n$, and $\Gamma(c)\sim \Gamma'(c)$ which holds since $\Gamma\equiv\Gamma'$.
	
	Apart from the \textsc{(Ident)} leaves, the rest of the derivation $\mathcal{D}$ also holds in the context $\Gamma'$. Hence we have obtained a derivation of $\Gamma'\vdash \Gamma(v)$.
\end{proof}

\begin{ext-prop}[Equivalence soundness]\label{prop:equiv-soundness}
	If $\Gamma\equiv\Gamma'$, then
	\begin{enumerate}
		\item[(a)] $\Gamma\vdash ck$ if and only if $\Gamma'\vdash ck$
		\item[(b)] $\Gamma\vdash e:ck$ if and only if $\Gamma'\vdash e:ck$
	\end{enumerate}
\end{ext-prop}
\begin{proof}
	It is sufficient to prove only the direct implication of both equivalences, since $\equiv$ is symmetric.
	We prove (a) and (b) simultaneously by induction on the
	derivation. In both cases, we start from a derivation $\mathcal{D}$ of $\Gamma\vdash ck$ (or $\Gamma\vdash e:ck$) and transform it into a derivation $\mathcal{D}'$ of $\Gamma'\vdash ck$ (or $\Gamma'\vdash e:ck$). 
	
	In $\mathcal{D}$, the only rules that read the context are \textsc{(Ident)} and \textsc{(Const)}. In those that do not, we can directly replace $\Gamma$ by $\Gamma'$.
	
	It remains to transform the \textsc{(Ident)} leaves of $\mathcal{D}$ (the case of \textsc{(Const)} is identical).
	Such a leaf concludes $\Gamma\vdash c:\Gamma(c)$. Under $\Gamma'$,
	the rule \textsc{(Ident)} gives $\Gamma'\vdash c:\Gamma'(c)$.
	Moreover, $\Gamma'\vdash\Gamma'(c)$ holds by consistency of
	$\Gamma'$, $\Gamma'\vdash\Gamma(c)$ holds by \lemmaref{cross-wf},
	and $\Gamma'(c)\sim\Gamma(c)$ holds by hypothesis. So we can use the rule \textsc{(Equiv-$\Gamma$)} to obtain that
	$\Gamma'\vdash\Gamma'(c)\equiv\Gamma(c)$, and \textsc{(Conv)}
	to conclude that $\Gamma'\vdash c:\Gamma(c)$, and thus replace the \textsc{(Ident)} rule.
\end{proof}

This last statement allows us to establish the consistency of a context by comparison with a context already known to be consistent.
\begin{ext-prop}[Consistency transfer]\label{prop:consistency-transfer}
	Let $\Gamma$ be a consistent context. If $\Gamma'$ is a context such that $dom(\Gamma)=dom(\Gamma')$ and, for all $v\in dom(\Gamma)$, $\Gamma(v)\sim\Gamma'(v)$ and $\Gamma\vdash \Gamma'(v)$.
	Then $\Gamma'$ is consistent, and so $\Gamma\equiv\Gamma'$.
\end{ext-prop}
\begin{proof}
	First, note that $\Gamma(v)\sim \Gamma'(v)$ implies they have the same rank.
	
	By strong induction on $n$, we prove that, for all $v\in dom(\Gamma)$ such that $\rank{\Gamma(v)}\leqslant n$, we have both
	\begin{itemize}
		\item[(i)] $\Gamma'\vdash \Gamma'(v)$
		\item[(ii)] $\Gamma'\vdash \Gamma(v)$
	\end{itemize}
	
	For $n=0$, then $\Gamma(v)$ is $\Tick$ or a clock
	variable, and $\Gamma'(v)\sim\Gamma(v)$ forces
	$\Gamma'(v)=\Gamma(v)$. Both judgments follow from a single \textsc{(Tick)} or \textsc{(Var)} rule.
	
	For $n\rightarrow n+1$. Consider $v$ such that $\rank{\Gamma(v)}=n+1$.
	
	For (i). By hypothesis, there is a derivation $\mathcal{D}$ of
	$\Gamma\vdash\Gamma'(v)$. Since $\Gamma$ is consistent, \lemmaref{derivation-structure}(a) applies to $\mathcal{D}$. Hence, every \textsc{(Ident)} leaf of	$\mathcal{D}$ with premise $(c, \Gamma(c))$ satisfies
	$\rank{\Gamma(c)}<n+1$. We turn $\mathcal{D}$ into a derivation of $\Gamma'\vdash\Gamma'(v)$.
	
	Every \textsc{(Ident)} leaf with conclusion $\Gamma\vdash c:\Gamma(c)$ is replaced by a derivation of $\Gamma'\vdash c:\Gamma(c)$. The latter is obtained by using the \textsc{(Conv)} rule. Its first premise is $\Gamma'\vdash c:\Gamma'(c)$, which holds by \textsc{(Ident)} since $\Gamma'$ has the same domain as $\Gamma$. The other premise is $\Gamma'\vdash \Gamma(c)\equiv \Gamma'(c)$, itself relying on three more premises, namely $\Gamma'\vdash \Gamma'(c)$, $\Gamma'\vdash \Gamma(c)$, and $\Gamma(c)\sim \Gamma'(c)$.
	The first two hold by induction hypothesis of both (i) and (ii), since $rank(\Gamma(c))\leq n$. The third one holds by hypothesis.
	
	In the rest of $\mathcal{D}$ it is sufficient to just replace $\Gamma$ by $\Gamma'$ as \textsc{(Ident)} is the only context-reading rule (constants cannot
	appear, as subsampling expressions are built from identifiers only).
	The resulting tree is therefore a derivation of
	$\Gamma'\vdash\Gamma'(v)$.
	
	For (ii), by consistency of $\Gamma$, there is a derivation of $\Gamma\vdash\Gamma(v)$, whose \textsc{(Ident)} leaves also satisfy
	$\rank{\Gamma(c)}\leqslant n$ by \lemmaref{derivation-structure}(a). The same
	transformation yields a derivation of $\Gamma'\vdash\Gamma(v)$.
	
	Since every clock has finite rank, point (i) holds for every $v\in dom(\Gamma)$, i.e. $\Gamma'$ is consistent.
\end{proof}

\fi

\subsection{Correctness theorems}\label{sec:correctness}
Let us now state the correctness of our clock inference.
We show it is complete, sound, and has the principal typing property.
To do this, we start by formalizing typability, the driving motivation
behind inference.
\begin{definition}
	We say that an expression $e$ is \emph{typable} (or
	\emph{clockable}) in a consistent context $\Gamma$ if and only if
	there is a substitution $\sigma$ that is valid under $\Gamma$ and a
	clock $ck$ such that $\sigma(\Gamma)\vdash e:ck$. For other
	fragments $f$ of the language, we say that $f$ is \emph{typable} (or
	\emph{clockable}) in a context $\Gamma$ if there is a substitution
	$\sigma$ that is valid under $\Gamma$ and such that
	$\sigma(\Gamma)\vdash f$.
\end{definition}

\ifextended
We are now ready to prove the correctness theorems.
\begin{ext-theorem}[Correctness of unification]\label{thm:unif-corr}
	Let $\Gamma$ be a consistent context and $ck_1,ck_2$ two clocks that
	are unifiable under $\Gamma$ and such that $\Gamma\vdash ck_1$ and
	$\Gamma\vdash ck_2$. Then, the unification procedure does not fail on
	$ck_1\doteq ck_2$ \textbf{(completeness)} and it yields a
	substitution $\delta$ (valid under $\Gamma$) satisfying:
	\begin{itemize}
		\item $\delta$ unifies $ck_1\doteq ck_2$ under $\Gamma$ (\textbf{soundness})
		\item for all $\tau$ unifier of $ck_1\doteq ck_2$ under $\Gamma$ we
		have $\Gamma\vdash \delta\lesssim\tau$ (\textbf{Most General Unifier})
	\end{itemize}
\end{ext-theorem}
\begin{proof}
	We proceed by case analysis on $ck_1$ and $ck_2$
	
	\textbf{If $base(ck_1)=base(ck_2)$}, then we must have $subs(ck_1)=subs(ck_2)$, otherwise the clocks would not be unifiable. In this case, unification yields $\delta=\id$ which is a most general unifier.
	
	\textbf{If $base(ck_1)=\var{x}$, and $base(ck_2)=\Tick$ (and symmetrically)}. Then we must have $subs(ck_1)\subseteq subs(ck_2)$, otherwise the clocks would not be unifiable. There are two cases. 
	
	\textbf{The first one assumes $ck_1=\var{x}$}, in which case unification yields $\delta=[\var{x}\mapsto ck_2]$. 
	We show that $\delta$ is valid under $\Gamma$ and that it is the most general unifier, i.e. for any competing unifier $\tau$, valid under $\Gamma$, we have $\Gamma\vdash \delta\lesssim \tau$, i.e. there exists $\gamma$, valid under $\delta(\Gamma)$, and such that $\gamma\circ\delta(\Gamma) \equiv \tau(\Gamma)$.

	\textbf{Step 1: $\delta$ is valid under $\Gamma$.}
	Since $base(ck_2)\neq \var{x}$ we have that $\delta$ is idempotent, and since $\Gamma\vdash ck_2$ by hypothesis, then $\delta$ is also well-formed under $\Gamma$. By \lemmaref{wf-implies-valid}, we have that $\delta$ is valid under $\Gamma$.
	
	\textbf{Step 2: $\Delta\equiv\Delta'$.}
	We now show that $\gamma=\tau|_{dom(\tau)\setminus\{\var{x}\}}$ is the witness we are looking for. By construction we have that $\gamma\circ\delta\sim\tau$.
	We denote $\tau(\Gamma)$ by $\Delta$ and $\gamma\circ\delta(\Gamma)$ by $\Delta'$.
	
	We wish to apply \propref{consistency-transfer} to obtain that $\Delta\equiv\Delta'$. To this end we must show that $\forall v\in dom(\Gamma), \Delta\vdash \Delta'(v)$.
	
	Let $v\in dom(\Gamma)$, we write $\Gamma(v)=\on{\on{\on{\beta}{s_1}}{\ldots}}{s_k}$ for $k\geq 0$.
	Note that, since $\tau$ is a unifier, we have $\tau(\Gamma)\vdash \tau(\var{x})\equiv\tau(ck_2)$.
	If $\beta\neq\var{x}$ then $\Delta(v)=\Delta'(v)$, and so by consistency of $\Delta$ (since $\tau$ is valid under $\Gamma$) we have $\Delta\vdash\Delta'(v)$.
	
	If $\beta=\var{x}$ then we have that
	\begin{align*}
		\delta(\Gamma(v))&=\on{\on{\on{\delta(\var{x})}{s_1}}{\ldots}}{s_k}\\
		&=\on{\on{\on{ck_2}{s_1}}{\ldots}}{s_k}
	\end{align*}
	Thus we have 
	\begin{align*}
		\Delta'(v)&=\on{\on{\on{\tau(ck_2)}{s_1}}{\ldots}}{s_k}\\ \Delta(v)&=\on{\on{\on{\tau(\var{x})}{s_1}}{\ldots}}{s_k}
	\end{align*}
	
	Define  $Q_i=\on{\on{\on{\tau(ck_2)}{s_1}}{\ldots}}{s_i}$ and $P_i = \on{\on{\on{\tau(\var{x})}{s_1}}{\ldots}}{s_i}$, for $i\in\{0, \ldots, k\}$.
	In particular, $Q_0=\tau(ck_2)$, $P_0 = \tau(\var{x})$, $Q_k=\Delta'(v)$ and $P_k=\Delta(v)$.
	
	We want to prove that $\Delta\vdash Q_k$.
	By finite induction we show that $\Delta\vdash Q_i\equiv P_i$ and that $\Delta\vdash s_{i+1}: Q_i$. Note that, by \lemmaref{prefix-wf} we have that $\Delta\vdash s_{i+1}:P_i$ for all $i\leq k$. And since $\tau$ is a unifier, we have $P_i\sim Q_i$ for all $i$.
	For $i=0$, $\Delta\vdash P_0\equiv Q_0$ stems from the fact that $\tau$ is a unifier. We can then apply a \textsc{(Conv)} rule to $\Delta\vdash P_0$ to obtain $\Delta\vdash s_1: Q_0$.
	For $i\rightarrow i+1$ we can use the rule \textsc{(On)} on $\Delta\vdash s_{i+1}:Q_i$ to obtain that $\Delta\vdash Q_{i+1}$ and thus $\Delta\vdash Q_{i+1}\equiv P_{i+1}$ by \textsc{(Equiv-$\Gamma$)}. From there, we can apply \textsc{(Conv)} to $\Delta\vdash s_{i+2}:P_{i+1}$ to deduce $\Delta\vdash s_{i+2}:Q_{i+1}$.
	
	In the end, we obtain $\Delta\vdash Q_k\equiv P_k$ and so in particular $\Delta\vdash Q_k$, i.e. $\Delta\vdash \Delta'(v)$. We can apply \propref{consistency-transfer} to get $\Delta\equiv\Delta'$, and by \lemmaref{equiv-subs-consistent-context} we get $\Gamma\vdash \tau\equiv\gamma\circ\delta$.
	
	\textbf{Step 3: $\gamma$ is valid under $\delta(\Gamma)$.}
	Let $\var{z}\in dom(\gamma)$. We must show $\gamma(\delta(\Gamma))\vdash \gamma(\var{z})$, i.e. $\Delta'\vdash \tau(\var{z})$. Validity of $\tau$ gives $\Delta\vdash \tau(\var{z})$, and by \propref{equiv-soundness}, since $\Delta\equiv\Delta'$, we have the desired result. And thus $\Gamma\vdash\delta\lesssim\tau$.

	\textbf{In the second case $ck_1$ is of the form $\on{\on{\on{\var{x}}{e_1}}{\ldots}}{e_n}$} where $n\geqslant 1$. We show this situation never occurs.
	Indeed, $n\geqslant 1$ implies $subs(ck_1)\neq\varnothing$, and since $subs(ck_1)\subseteq subs(ck_2)$, we have that $subs(ck_1)\cap subs(ck_2)\neq\varnothing$. Let $d$ be in this intersection. By using \lemmaref{ck-sub} on both clock we obtain that $base(ck_1) = base(\Gamma(|d|)) = base(ck_2)$ a contradiction.
	
	\textbf{If $base(ck_1)=\var{x}$ and $base(ck_2)=\var{y}$, for some distinct variables $\var{x}$ and $\var{y}$.} There are two cases.
	
	\textbf{The first one is $ck_1=\var{x}$ or $ck_2=\var{y}$.} Without loss of generality suppose $ck_1=\var{x}$. Unification yields $\delta=[\var{x}\mapsto ck_2]$, and we show that $\delta$ is valid under $\Gamma$ and that it is the most general unifier, i.e. for any competing unifier $\tau$, valid under $\Gamma$, we have $\Gamma\vdash \delta\lesssim \tau$, i.e. there exists $\gamma$, valid under $\delta(\Gamma)$ and such that $\gamma\circ\delta(\Gamma) \equiv \tau(\Gamma)$.
	This case is proven exactly as the previous one.
	
	\textbf{The second one is when $subs(ck_1)\neq\varnothing$ and $subs(ck_2)\neq\varnothing$.} We show this situation never occurs.
	Suppose there is a unifier $\tau$. First we write $ck_1=\on{\on{\on{\var{x}}{e_1}}{\ldots}}{e_n}$ and $ck_2=\on{\on{\on{\var{y}}{f_1}}{\ldots}}{f_m}$. By hypothesis $n\geqslant 1$ and $m\geqslant 1$. Note that $subs(ck_1)\cap subs(ck_2)=\varnothing$, otherwise, by \lemmaref{ck-sub}, a common subsampling would have both base $\var{x}$ and $\var{y}$, impossible.
	
	Let $c=|ebase(e_1)|$ and $d=|ebase(f_1)|$. The idea of the proof is that after unification, there should be some circular dependency on $c$ and $d$.
	Since $\Gamma\vdash ck_1$, we have that $\Gamma\vdash e_1:\var{x}$, and similarly $\Gamma\vdash ck_2$ implies $\Gamma\vdash f_1:\var{y}$.
	By \lemmaref{ck-sub} we have that $\Gamma(c) = \var{x}$ and $\Gamma(d) = \var{y}$.
	
	Since $\tau$ is a unifier, we have $subs(\tau(ck_1))=subs(\tau(ck_2))$. Moreover
	\[
	\begin{aligned}
		subs(\tau(ck_1))&=subs(\tau(\var{x}))\cup\{ebase(e_1), \ldots, ebase(e_n)\}\\
		subs(\tau(ck_2))&=subs(\tau(\var{y}))\cup\{ebase(f_1), \ldots, ebase(f_m)\}
	\end{aligned}
	\]
	With the same argument given above for $subs(ck_1)\cap subs(ck_2)=\varnothing$ we have that $ebase(e_1)\in subs(\tau(\var{y}))$ and $ebase(f_1)\in subs(\tau(\var{x}))$. Now write
	\begin{align*}
		\tau(\var{x})&=\on{\on{\on{\beta}{r_1}}{\ldots}}{r_k}\\
		\tau(\var{y})&=\on{\on{\on{\beta}{s_1}}{\ldots}}{s_\ell}
	\end{align*}
	where $\beta$ is their common base (either $\Tick$ or some clock variable) and $k, \ell\in \NN$.
	
	Let $j$ be the first index at which $ebase(r_j)=ebase(f_1)$.
	Since $\tau$ is valid under $\Gamma$ we have $\tau(\Gamma)\vdash \tau(\var{x})$ and so $\tau(\Gamma)\vdash r_j:\alpha$ where $\alpha=\on{\on{\on{\beta}{r_1}}{\ldots}}{r_{j-1}}$.
	By applying \lemmaref{ck-sub} we obtain that $subs(\tau(\Gamma)(d))\subseteq subs(\alpha)=\{ebase(r_1), \ldots, ebase(r_{j-1})\}$.
	And since $\Gamma(d)=\var{y}$ we have $\tau(\Gamma)(d)=\tau(\var{y})$ which gives $subs(\tau(\var{y}))\subseteq \{ebase(r_1), \ldots, ebase(r_{j-1})\}$.
	
	Similarly for $\tau(\var{x})$ we obtain that $subs(\tau(\var{x}))\subseteq \{ebase(s_1), \ldots, ebase(s_{i-1})\}$ for some $i$ that is the first index at which $ebase(s_i)=ebase(e_1)$.
	
	Now we can establish
	$$
	\begin{aligned}
		ebase(f_1)&\in subs(\tau(\var{x})) \\
		&\subseteq \{ebase(s_1), \ldots, ebase(s_{i-1})\} \\
		&\subseteq subs(\tau(\var{y})) \\
		&\subseteq \{ebase(r_1), \ldots, ebase(r_{j-1})\}
	\end{aligned}
	$$
	Which is a contradiction, since, by minimality of $j$, we have $ebase(f_1)\notin\{ebase(r_1), \ldots, ebase(r_{j-1})\}$.
\end{proof}
\fi

\begin{theorem}[Expression inference correctness]\label{thm:expr-corr}
	Let $e$ be an expression and $\Gamma$ a consistent context in which $e$ is typable. Then inference yields \textbf{(completeness)} a clock $ck$ and a substitution $\delta$ (valid under $\Gamma$) such that
	\begin{itemize}
		\item $\delta(\Gamma)\vdash e: ck$ (\textbf{soundness})
		\item for every substitution $\tau$ valid under $\Gamma$ and witnessing typing for $e$, we have $\Gamma\vdash\delta\lesssim\tau$ (\textbf{principality})
	\end{itemize}
\end{theorem}
\ifextended
\begin{proof}
	By structural induction on $e$. The base cases ($e = x$ and $e = k$) are immediate. Inference returns the identity substitution, and principality holds trivially for $\id$.
	
	All inductive cases follow a uniform scheme. Consider $e=\fby{k}{e'}$ as representative. Inference processes $k$ in $\Gamma$, yielding $(\sigma_1, ck_1)$ such that $\sigma_1(\Gamma)\vdash k:ck_1$, then $e'$ in $\sigma_1(\Gamma)$, yielding $(\sigma_2, ck_2)$ such that $\sigma_2(\sigma_1(\Gamma))\vdash e':ck_2$, and finally unifies $\sigma_2(ck_1)$ with $ck_2$ under $\sigma_2(\sigma_1(\Gamma))$, yielding $\delta_0$. The returned substitution is $\delta = \delta_0 \circ \sigma_2 \circ \sigma_1$.
	
	For \textbf{soundness}. By \lemmaref{subst-preservation}, each factor of $\delta$ preserves the typings established in earlier stages, and the unification supplies the equivalence needed for \textsc{(Conv)}.
	
	For \textbf{principality}. 	
	Given a competing $\tau$ valid over $\Gamma$ such that $\tau(\Gamma) \vdash e : ck$, \lemmaref{inversion} yields $\tau(\Gamma)\vdash k:ck'$ and $\tau(\Gamma) \vdash e' : ck'$ for each subexpression. We show $\Gamma\vdash \delta\lesssim \tau$. The induction hypothesis on $k$ gives $\Gamma \vdash \sigma_1 \lesssim \tau$ with witness $\gamma_1$ valid under $\sigma_1(\Gamma)$. By \lemmaref{equiv-subs-consistent-context}, $\tau(\Gamma) \equiv \gamma_1(\sigma_1(\Gamma))$.
	We transfer the typing of $e'$ to this equivalent context via \propref{equiv-soundness}, so that $\gamma_1$ witnesses typing of $e'$ in $\sigma_1(\Gamma)$. 
	Induction hypothesis on $e'$ yields $\sigma_1(\Gamma) \vdash \sigma_2 \lesssim \gamma_1$ with witness $\gamma_2$ valid under $\sigma_2(\sigma_1(\Gamma))$.
	
	It remains to show that $\sigma_2(\sigma_1(\Gamma))\vdash \delta_0\lesssim \gamma_2$.
	To do so, we prove that $\gamma_2$ unifies $\sigma_2(ck_1)$ and $ck_2$.
	By \lemmaref{subst-preservation} applied to $\gamma_2$, we have
	$\gamma_2(\sigma_2(\sigma_1(\Gamma))) \vdash k : \gamma_2(\sigma_2(ck_1))$
	and $\gamma_2(\sigma_2(\sigma_1(\Gamma))) \vdash e' : \gamma_2(ck_2)$.
	By \lemmaref{equiv-subs-consistent-context},
	$\gamma_2(\sigma_2(\sigma_1(\Gamma))) \equiv \gamma_1(\sigma_1(\Gamma))
	\equiv \tau(\Gamma)$,
	and under this last context, both $k$ and $e'$ have clock $ck'$.
	By \propref{equiv-soundness} and uniqueness of typing, the transferred clocks
	are equivalent, yielding
	$\gamma_2(\sigma_2(\sigma_1(\Gamma))) \vdash
	\gamma_2(\sigma_2(ck_1)) \equiv \gamma_2(ck_2)$.
	
	By the Most General Unifier property of \thmref{unif-corr} we have
	$\sigma_2(\sigma_1(\Gamma)) \vdash \delta_0 \lesssim \gamma_2$
	with witness $\gamma_3$ valid under $\delta(\Gamma)$.
	Under $\Gamma$, composing the three factorizations yields
	$$
	\tau
	\equiv \gamma_1 \circ \sigma_1
	\equiv \gamma_2 \circ \sigma_2 \circ \sigma_1
	\equiv \gamma_3 \circ \delta_0 \circ \sigma_2 \circ \sigma_1
	= \gamma_3 \circ \delta
	$$
	We conclude that $\Gamma \vdash \delta \lesssim \tau$.
	
	The \textbf{when} case follows the same scheme; the only specificity is that the returned clock has the form $\on{\delta_0(ck_s)}{sse}$, whose well-formedness in $\delta(\Gamma)$ follows from \textsc{(On)} applied to the inferred typing of $sse$. The \textbf{merge} case involves two unification steps instead of one but is otherwise identical.
\end{proof}
\fi

The correctness results established for expressions directly extend to the clocking of entire nodes, as a node is defined compositionally over its constituent expressions.
\begin{coro}[Correctness of inference]
	Let $p$ be a node and $\Gamma_0$ its initial context. If $p$ is typable, then the inference procedure does not fail \textbf{(completeness)}, and it yields a substitution $\sigma$, valid over $\Gamma_0$, such that 
	\begin{itemize}
		\item $\sigma(\Gamma_0)\vdash p$ (\textbf{soundness})
		\item for every substitution $\tau$ valid over $\Gamma_0$ and such that $\tau(\Gamma_0)\vdash p$, we have $\Gamma_0\vdash \sigma\lesssim \tau$ (\textbf{principality})
	\end{itemize}
\end{coro}

\section{Conclusion\label{sec:conclusion}}
While dataflow principles are pervasive in ML languages and frameworks, the support for data-dependent control remains limited in practice.
Symmetrically in the domain of embedded systems, the clock calculi of dataflow synchronous languages address many of the issues with data-dependent control.
These calculi offer strong guarantees, such as deadlock freedom and static resource allocation.
Unfortunately, backpropagation training, gated mixture of experts, and
other important ML algorithms are not modeled efficiently by
state-of-the-art clock calculi.
This paper extends the Lustre clock calculus to lift this hurdle.
In particular, we allow subsampling Boolean conditions
to commute whenever the presence of data allows it.
The clock congruence of \ccnew\ reorganizes the Lustre clock tree
into a directed acyclic graph, providing a clean definition of clock intersection.
This results in efficient real-time resource allocation without the need for a SAT solver, with a radical improvement of both explainability---the extraction of meaningful error diagnostics---and computational complexity.
We proved the correctness of the new calculus.
We also implemented it in a prototype compiler, powerful enough to implement training algorithms, GMoE, pipelined execution, and reinforcement learning.





Future work includes showcasing the language expressiveness and
efficient code generation on real-world ML applications.
We are particularly interested in concrete illustrations of
how dataflow synchrony facilitates debugging,
correctness proofs and static resource allocation.

\bibliographystyle{IEEEtran} 
\bibliography{biblio}

@misc{onnx,
    author = {Junjie Bai and Fang Lu and Ke Zhang et al.},
    title = {ONNX: Open Neural Network Exchange},
    year = {2019},
    publisher = {GitHub},
    journal = {GitHub repository},
    howpublished = {\url{https://github.com/onnx/onnx}}
}

@inproceedings{girault,
author = {A. Girault and X. Nicollin},
year = {2003},
title = {Clock-Driven Automatic Distribution of {Lustre} Programs},
booktitle = {EMSOFT}
}

@ARTICLE{10753468,
  author={T. Bourke and M. Pouzet},
  journal={IEEE Embedded Systems Letters}, 
  title={Lustre, Fast First, and Fresh}, 
  year={2025},
  volume={17},
  number={2},
  }

@article{hugo,
	author = {H. Pompougnac and U. Beaugnon and A. Cohen and D. Potop Butucaru},
	title = {Weaving Synchronous Reactions into the Fabric of {SSA}-form Compilers},
	year = {2022},
	volume = {19},
	number = {2},
	journal = {ACM Trans. Archit. Code Optim.}
}

@Article{endochrony,
	author={Potop-Butucaru, D. and Caillaud, B.},
	title={Correct-by-construction asynchronous implementation of modular
	synchronous specifications},
	journal={Fundamenta Informaticae},
	year={2006},
	volume={CXXVII}
}

@article{multiproc-part,
	title={From Dataflow Specification to Multiprocessor Partitioned Time-triggered Real-time Implementation},
	volume={2},
	journal={Leibniz Transactions on Embedded Systems},
	author={Carle, Thomas and Potop-Butucaru, Dumitru and Sorel, Yves and Lesens, David},
	year={2015},
}

@inproceedings{from-simu,
	author = {Caspi, Paul and Curic, Adrian and Maignan, Aude and Sofronis, Christos and Tripakis, Stavros and Niebert, Peter},
	title = {From {Simulink} to {SCADE}/{Lustre} to {TTA}: A Layered Approach for Distributed Embedded Applications},
	year = {2003},
	booktitle = {LCTES}
}

@InProceedings{lustreRTSS,
  author = 	 {J.L. Bergerand and P. Caspi and D. Pilaud and N. Halbwachs
                  and E. Pilaud},
  title = 	 {Outline of a Real Time Data Flow Language},
  booktitle =    {Proceedings RTSS},
  year = 	 {1985}
}

@inproceedings{Lustre,
author = {Caspi, P. and Pilaud, D. and Halbwachs, N. and Plaice, J. A.},
title = {LUSTRE: a declarative language for real-time programming},
year = {1987},
isbn = {0897912152},
publisher = {Association for Computing Machinery},
address = {New York, NY, USA},
url = {https://doi.org/10.1145/41625.41641},
doi = {10.1145/41625.41641},
booktitle = {Proceedings of the 14th ACM SIGACT-SIGPLAN Symposium on Principles of Programming Languages},
pages = {178–188},
numpages = {11},
location = {Munich, West Germany},
series = {POPL '87}
}

@inproceedings{Coh06,
	author = {Cohen, Albert and Duranton, Marc and Eisenbeis, Christine and Pagetti, Claire and Plateau, Florence and Pouzet, Marc},
	title = {N-synchronous {Kahn} networks: a relaxed model of synchrony for real-time systems},
	year = {2006},
	publisher = {Association for Computing Machinery},
	booktitle = {POPL}
}

@inproceedings{shazeer2017outrageously,
	title={Outrageously Large Neural Networks: The Sparsely-Gated Mixture-of-Experts Layer},
	author={N. Shazeer and A. Mirhoseini and K. Maziarz and A. Davis and Q. Le and G. Hinton and J. Dean},
	booktitle={ICLR},
	year={2017}
}

@book{astrom1997computer,
	title={Computer-Controlled Systems: Theory and Design},
	author={K. J. {\AA}str{\"o}m and B. Wittenmark},
	year={1997},
	publisher={Prentice Hall}
}

@inproceedings{gu2024mamba,
	title={Mamba: Linear-Time Sequence Modeling with Selective State Spaces},
	author={Gu, Albert and Dao, Tri},
	booktitle={COLM},
	year={2024}
}

@inproceedings{rybakov2020streaming,
	title={Streaming Keyword Spotting on Mobile Devices},
	author={O. Rybakov and N. Kononenko and N. Adam and S. Subrahmanya and F. Beaufays},
	booktitle={Interspeech},
	year={2020},
}

@book{goodfellow2016deep,
	title={Deep Learning},
	author={I. Goodfellow and Y. Bengio and A. Courville},
	publisher={MIT Press},
	year={2016}
}

@InProceedings{lushist,
	author = 	 {N. Halbwachs},
	title = 	 {A synchronous language at work: the story of {Lustre}},
	booktitle = {MEMOCODE},
	year = 	 {2005}
}

@inproceedings{emsoft09,
	author = {D. Potop-Butucaru and R. de Simone and Y. Sorel and J.-P. Talpin},
	title = {Clock-driven distributed real-time implementation of endochronous synchronous programs},
	year = {2009},
	booktitle = {Proceedings EMSOFT}
}

@inproceedings{cai2016tensorflow,
	title        = {TensorFlow Debugger: Debugging Dataflow Graphs for Machine Learning},
	author       = {Shanqing Cai and Eric Breck and Eric Nielsen and Michael Salib and D. Sculley},
	booktitle    = {Proceedings of the Reliable Machine Learning in the Wild Workshop at NIPS},
	year         = {2016}
}

@article{tambon2021silent,
	title={Silent bugs in deep learning frameworks: an empirical study of Keras and TensorFlow},
	author={Tambon, Florian and Nikanjam, Amin and Lavall{\'e}e, Foutse and Antoniol, Giuliano},
	journal={Empirical Software Engineering},
	volume={26},
	number={4},
	year={2021}
}

@Misc{xla,
	author = 	 {Sanjoy Das},
	title = 	 {Control Flow in {TensorFlow} {\&} {XLA}'s Auto-Clustering},
	howpublished = {TensorFlow dev team blog post},
	month = 	 {oct},
	year = 	 {2018},
	note = 	 {\url{https://www.playingwithpointers.com/blog/control-flow-and-tf-xla.html}}
}

@inproceedings{GeilenBasten2010KPN,
	title     = {Kahn Process Networks and a Reactive Extension},
	author    = {M.C.W. Geilen and Twan Basten},
	editor    = {Shuvra S. Bhattacharyya and Ed F. Deprettere and Rainer Leupers and Jarmo Takala},
	booktitle = {Handbook of Signal Processing Systems},
	year      = {2010},
	publisher = {Springer}
}

@inproceedings{kahn1974semantics,
	title={The Semantics of a Simple Language for Parallel Programming},
	author={Kahn, Gilles},
	booktitle={Information Processing 74, Proceedings of the IFIP Congress 74},
	year={1974},
	publisher={North-Holland Publishing Co.},
}

@inproceedings{10.1145/1967677.1967688,
author = {Gamatie, Abdoulaye and Gonnord, Laure},
title = {Static analysis of synchronous programs in signal for efficient design of multi-clocked embedded systems},
year = {2011},
publisher = {Association for Computing Machinery},
series = {LCTES}
}

@inproceedings{10.1145/207110.207134,
author = {Amagb\'{e}gnon, Pascalin and Besnard, Lo\"{\i}c and Le Guernic, Paul},
title = {Implementation of the data-flow synchronous language SIGNAL},
year = {1995},
publisher = {Association for Computing Machinery},
series = {PLDI}
}

@article{BENVENISTE1991103,
title = {Synchronous programming with events and relations: the SIGNAL language and its semantics},
journal = {Science of Computer Programming},
volume = {16},
number = {2},
year = {1991},
author = {Albert Benveniste and Paul {Le Guernic} and Christian Jacquemot},
}

@TechReport{tfgraph,
	author = 	 {The TensorFlow authors},
	title = 	 {Implementation of Control Flow in TensorFlow},
	year = 	 {2017},
	url = 	 {http://download.tensorflow.org/paper/white_paper_tf_control_flow_implementation_2017_11_1.pdf},
}

@inproceedings{yu2018dynamic,
	title={Dynamic control flow in large-scale machine learning},
	author={Yuan Yu, Yuan and Abadi, Martín and Barham, Paul and Brevdo, Eugene and Burrows, Mike and Davis, Andy and Dean, Jeff and Ghemawat, Sanjay and Harley, Tim and Hawkins, Peter and others},
	booktitle={EuroSys},
	year={2018}
}

@article{flink,
	author = {P. Carbone and A. Katsifodimos and S. Ewen and V. Markl and S. Haridi and K. Tsoumas},
	title = {Apache Flink: Stream and Batch Processing in a Single Engine},
	year = {2015},
	journal={IEEE Data(base) Engineering Bulletin},
	volume={36},
}

@Book{simulink,
	author = 	 {H. Klee and R. Allen},
	title = 	 {Simulation of Dynamic Systems with MATLAB and Simulink},
	publisher = 	 {CRC Press},
	year = 	 {2018}
}

@article{CaP96,
	author = {Caspi, Paul and Pouzet, Marc},
	title = {Synchronous {Kahn} {networks}},
	year = {1996},
	publisher = {Association for Computing Machinery},
	volume = {31},
	number = {6},
	journal = {SIGPLAN Not.},
}

@inproceedings{CoP03,
	author = {Cola\c{c}o, Jean-Louis and Pouzet, Marc},
	editor = {Alur, Rajeev and Lee, Insup},
	title = {Clocks as First Class Abstract Types},
	booktitle = {Embedded Software},
	year = {2003}
}

@article{Bal08,
	author = {Biernacki, Dariusz and Cola\c{c}o, Jean-Louis and Hamon, Gregoire and Pouzet, Marc},
	title = {Clock-directed modular code generation for synchronous data-flow languages},
	year = {2008},
	volume = {43},
	number = {7},
	journal = {SIGPLAN Not.}
}

@incollection{Ba01,
	title = {Chapter 8 - {Unification Theory}},
	editor = {Alan Robinson and Andrei Voronkov},
	booktitle = {Handbook of Automated Reasoning},
	year = {2001},
	author = {Franz Baader and Wayne Snyder and Paliath Narendran and Manfred Schmidt-Schauss and Klaus Schulz}
}

@inproceedings{DamasMilner82,
	author = {Damas, Luis and Milner, Robin},
	title = {Principal type-schemes for functional programs},
	year = {1982},
	publisher = {ACM SIGPLAN-SIGACT},
	booktitle = {POPL '82}
}

\end{document}